\documentclass[a4paper,11pt]{article}

\pdfoutput=1

\usepackage{fullpage}

\usepackage{amssymb,amsmath,amsfonts,amsthm}

\newtheorem{theorem}{Theorem}[section]

\newtheorem{corollary}[theorem]{Corollary}
\newtheorem*{claim}{Claim}
\theoremstyle{definition}
\newtheorem{definition}[theorem]{Definition}
\newtheorem{problem}[theorem]{Problem}

\usepackage{tikz}

\usetikzlibrary{external}
\tikzset{external/system call={lualatex \tikzexternalcheckshellescape -halt-on-error -interaction=batchmode -jobname "\image" "\texsource"}}

\newenvironment{parameterizedproblem}%
  {%
    \leavevmode\nobreak\par
    \begin{list}%
      {}%
      {%
        \def\labelstyle{\itshape}
        \setlength{\topsep}{0pt}%
        \settowidth{\labelwidth}{\labelstyle Parameter:}%
        \setlength{\leftmargin}{\labelwidth}%
        \addtolength{\leftmargin}{\labelsep}%
        \setlength{\itemsep}{0pt}%
        \setlength{\parsep}{0pt}%
      }%
      \begin{samepage}%
        \def\instance{\item[\labelstyle Instance:]}%
        \def\parameter{\item[\labelstyle Parameter:]}%
        \def\question{\item[\labelstyle Question:]}%
  }%
  {%
    \end{samepage}%
    \end{list}%
  }

\newcommand\Class[1]{\mathchoice{\text{\normalfont\small$\mathrm{#1}$}}{\text{\normalfont\small$\mathrm{#1}$}}{\text{\normalfont$\mathrm{#1}$}}{\text{\normalfont$\mathrm{#1}$}}}     
\newcommand{\Lang}[1]{\ifmmode{\text{\textsc{#1}}}\else\textsc{#1}\fi}
\newcommand\PLang[1][]{p\def\test{#1}\ifx\test\stockhustantauempty\else_{\mathrm{#1}}\fi\text-\penalty15\Lang}
\def\stockhustantauempty{}

\newcommand\Para[1][\text-]{\Class{para#1}}

\newcommand\PLangText[1][]{$p\def\test{#1}\ifx\test\empty\else_{\mathrm{#1}}\penalty15\fi$-\penalty15\hskip0pt\textsc}

\title{%
  Completeness Results for Parameterized~Space~Classes
}
\author{%
  Christoph Stockhusen\and
  Till Tantau
}
\date{%
  Institute for Theoretical Computer Science\\
  Universit\"at zu L\"ubeck\\
  D-23538 L\"ubeck, Germany\\
  \texttt{\{stockhus,tantau\}@tcs.uni-luebeck.de}
}

\begin{document}
  
\maketitle

\begin{abstract}
  The parameterized complexity of a problem is generally considered
  ``settled'' once it has been shown to lie in FPT or to be complete
  for a class in the W-hierarchy or a similar parameterized hierarchy.
  Several natural parameterized problems have, however, resisted such
  a classification. At least in some cases, the reason is that upper
  and lower bounds for their parameterized \emph{space} complexity
  have recently been obtained that rule out completeness results for
  parameterized \emph{time} classes. In this paper, we make progress
  in this direction by proving that the associative generability
  problem and the longest common subsequence problem  are complete for
  parameterized space classes. These classes are defined in terms of
  different forms of bounded nondeterminism and in terms of
  simultaneous time--space bounds.  As a technical tool we introduce a
  ``union operation'' that translates between problems complete for
  classical complexity classes and for W-classes.
\end{abstract}

\section{Introduction}

Parameterization has become a powerful paradigm in complexity theory,
both in theory and practice. Instead of just considering the runtime
of an algorithm as a function of the input \emph{length}, we analyse
the runtime as a multivariate function depending on a number of
different problem \emph{parameters}, the input length being just one
of them. While in classical complexity theory instead of ``runtime''
many other resource bounds have been studied in great detail, in the
parameterized world the focus has lain almost entirely on time
complexity. This changed when in a number of
papers~\cite{Caietal1997,ElberfeldST2012,Guillemot2011} it was 
shown for different natural problems, including the
vertex cover problem, the feedback vertex set problem, and the longest
common subsequence problem, that their parameterized \emph{space}
complexity is of interest. Indeed, the parameterized space
complexity of natural problems can explain why some problems in
$\Class{FPT}$ are easier to solve than others (namely, because they
lie in much smaller space classes) and why some problems cannot be
classified as complete for levels of the weft hierarchy (namely,
because upper and lower bounds on their space complexity rule out such
completeness results unless unlikely collapses occur).

\paragraph{Our Contributions.}

In the present paper, we present completeness results of natural
parameterized problems for different parameterized space complexity classes. The
classes we study are of two kinds: First, parameterized classes of
\emph{bounded nondeterminism} and, second, parameterized classes where
the \emph{space and time} resources of the machines are bounded
\emph{simultaneously.} In both cases, we introduce the classes  
for systematic reasons, but also because they are needed to classify
the complexity of the natural problems that we are interested in.

In the context of bounded nondeterminism, we introduce a general
``union operation'' that turns any language into a parameterized
problem in such a way that completeness of the language for some
complexity class $C$ carries over to completeness of the parameterized
problem for a class ``$\Para[W]C$,'' which we will define rigorously
later. Building on this result, we show that many union versions of
graph problems are complete for $\Para[W]\Class{NL}$ and
$\Para[W]\Class L$, but the theorem can also be used to show that
$\PLang{weighted-sat}$ is complete for $\Para[W]\Class{NC}^1$. Our
technically most challenging result is that the associative
generability problem parameterized by the generator set size is
complete for the class $\Para[WN]\Class{L}$.

Regarding time--space classes, we present different problems that are
complete for the class of problems solvable ``nondeterministically in
fixed-parameter time and slice-wise logarithmic space.'' Among these
problems are the longest common subsequence problem parameterized by
the number of strings, but also the acceptance problem for certain
cellular automata parameterized by the number of cells and also a
simple but useful pebble game.

\paragraph{Related Work.}

Early work on parameterized space classes is due to Cai et
al.~\cite{Caietal1997} who introduced the classes $\Para\Class L$ 
and $\Para\Class{NL}$, albeit under different names, and showed that
several important problems in $\Class{FPT}$ lie in these classes: the
parameterized vertex cover problem lies in $\Para\Class L$ and the 
parameterized $k$-leaf spanning tree problem lies in
$\Para\Class{NL}$. Later, Flum and Grohe~\cite{FlumG2003} showed that
the parameterized model checking problem of first-order formulas on
graphs of bounded degree lies in~$\Para\Class L$. In particular,
standard parameterized graph problems belong to $\Para\Class L$ when
we restrict attention to bounded-degree graphs. Recently,
Guillemot~\cite{Guillemot2011} showed that the longest common
subsequence problem (\textsc{lcs}) is equivalent under fpt-reductions
to the short halting problem for \textsc{ntm}s, where the time and
space bounds are part of the input and the space bound is the
parameter. Our results differ from Guillemot's insofar as we use
weaker reductions (para-$\Class L$- rather than fpt-reductions) and
prove completeness for a class defined using a machine model rather
than for a class defined as a reduction closure. The paper
\cite{ElberfeldST2012} by Elberfeld and us  is similar to the present
paper insofar as it also introduces new parameterized space complexity
classes and presents upper and lower bounds for natural parameterized
problems. The core difference is that in the present paper we focus on
completeness results for natural problems rather than ``just'' on
upper and lower bounds.

\paragraph{Organisation of This Paper.}

In Section~\ref{classes} we review the parameterized space
classes previously studied in the literature and introduce some new
classes that will be needed in the later sections. For some of the
classes from the literature we propose new names in order to
systematise the naming and to make connections between the different
classes easier to spot. In Section~\ref{section-bounded-non} we study
problems complete for classes defined in terms of bounded
nondeterminism, in Section~\ref{section-space-time} we do the same for
time--space classes. 

Due to lack of space, all proofs have been omitted. They can be found
in the technical report version of this paper.

\section{Parameterized Space Classes}

\label{classes}

Before we turn our attention to parameterized \emph{space} classes, let
us first review some basic terminology. As 
in~\cite{ElberfeldST2012}, we define a \emph{parameterized problem} as
a pair~$(Q,\kappa)$ of a language $Q \subseteq \Sigma^*$ and a
parameterization $\kappa\colon\Sigma^*\to\mathbb N$ that maps input
instances to parameter values and that 
is computable in logarithmic space.\footnote{In the classical
definition, Downey and Fellows~\cite{DowneyF1999} just require the 
parameterization to be computable, while Flum and
Grohe~\cite{FlumG2006} require it to be computable in polynomial
time. Whenever the parameter is part of the input, it is certainly
computable in logarithmic space.} For a classical complexity
class~$C$, a parameterized problem $(Q,\kappa)$ belongs to the
\emph{para-class} $\Para C$ if there are an alphabet~$\Pi$, a
computable function $\pi\colon\mathbb N\to\Pi^*$, and a language
$A\subseteq\Sigma^*\times\Pi^*$ with $A\in C$ such that for all $x\in
\Sigma^*$ we have $x\in Q\iff\big(x,\pi\big(\kappa(x)\big)\big)\in A$.
The problem is in the \emph{X-class} $\Class{X}C$ if for every number
$w \in \mathbb N$ the slice $Q_w=\{\,x\mid\text{$x\in Q$ and
  $\kappa(x)=w$}\}$ lies in~$C$. It is immediate from the definition
that $\Para C \subseteq \Class XC$ holds.

The ``popular'' class $\Class{FPT}$ is the same as $\Para\Class P$. In
terms of the $O$-notation, a parameterized problem $(Q,\kappa)$ is in
$\Para\Class P$ if there is a function $f\colon\mathbb N\to\mathbb N$
such that the question $x\in Q$ can be decided within time
$f(\kappa(x)) \cdot |x|^{O(1)}$. By comparison, $(Q,\kappa)$ is in
$\Para\Class L$ if $x\in Q$ can be decided within space
$f(\kappa(x)) + O(\log|x|)$; and for $\Para\Class{PSPACE}$ the
space requirement is $f(\kappa(x)) \cdot |x|^{O(1)}$. The
class $\Class{XP}$ is in wide use in parameterized complexity theory;
the logarithmic space classes $\Class{XL}$ and $\Class{XNL}$ have
previously been studied by Chen et
al.~\hbox{\cite{FlumG2003,ChenFG2003}}.

To simplify the notation, let us write $f_x$~for
$f\big(\kappa(x)\big)$ and $n$~for $|x|$ in the following. Then the
time bound for $\Para\Class P$ can be written as $f_x n^{O(1)}$ and
the space bound for $\Para\Class L$ as $f_x + O(\log n)$. 

\emph{Parameterized logspace reductions} ($\Para\Class L$-reductions)
are the natural restriction of fpt-reductions to logarithmic space: A
$\Para\Class L$-reduction from a parameterized problem
$(Q_1,\kappa_1)$ to $(Q_2,\kappa_2)$ is a mapping
$r\colon\Sigma_1^*\to\Sigma_2^*$ such~that 
\begin{enumerate}
\item for all $x\in\Sigma_1^*$ we have $x\in Q_1\iff r(x)\in Q_2$,
\item $\kappa_2\big(r(x)\big)\leq g\big(\kappa_1(x)\big)$ for some
  computable function~$g$, and,
\item $r$ is $\Para\Class L$-computable with respect to $\kappa_1$
  (that is, there is a Turing machine that outputs $r(x)$ on input $x$
  and needs space at most $f(\kappa_1(x))+O(\log |x|)$ for some
  computable function $f$).  
\end{enumerate}
Using standard arguments one can show that all classes in this paper are 
closed with respect to $\Para\Class L$-reductions; with the possible
exception of $\Para[W]\Class{NC}^1$, a class we encounter in
Theorem~\ref{theorem-nc1}. Throughout this paper, all completeness and
hardness results are meant with respect to $\Para\Class L$-reductions.

\subsection{Parameterized Bounded Nondeterminism}

\label{section-def-bounded-nondet}

While the interplay of nondeterminism and parameterized space may seem
to be simple at first sight ($\Class{NL}$ is closed under complement
and $\Class{NPSPACE}$ is even equal to $\Class{PSPACE}$, so only
$\Class{XNL}$ and $\Para\Class{NL}$ appear interesting), a closer look
reveals that useful and interesting new classes arise when we bound
the amount of nondeterminism used by machines in dependence on the
parameter. For this, it is useful to view nondeterministic
computations as deterministic computations using ``choice tapes'' or
``tapes filled with nondeterministic bits.'' These are extra tapes for
a deterministic Turing machine, and an input word is accepted if there
is at least one bitstring that we can place on this extra tape at the
beginning of the computation such that the Turing machine accepts. It
is well known that $\Class{NP}$ and $\Class{NL}$ can be defined in
this way using deterministic polynomial-time or logarithmic-space
machines, respectively, that have \emph{one-way} access to a choice
tape. (For $\Class{NP}$ it makes no difference whether we have one- or
two-way access, but logspace \textsc{dtm}s with access to a two-way
choice tape can accept all of $\Class{NP}$.)

Classes of \emph{bounded} nondeterminism arise when we restrict the
length of the bitstrings on the choice tape. For instance, the classes
$\beta^h$ for $h \ge 1$, see~\cite{KintalaF1980} and
also~\cite{BussG1993} for variants, are defined in the same way as
$\Class{NP}$ above, only the length of the bitstring on the choice 
tape may be at most $O(\log^h n)$. Classes of \emph{parameterized
  bounded} nondeterminism arise when we restrict the length the
bitstring on the choice tape in dependence not only on the input
length, but also of the parameter. Furthermore, in the context of
bounded space computations, it also makes a difference whether we have
one-way or two-way access to the choice tapes.

\begin{definition}
  Let $C$ be a complexity class defined in terms of a deterministic
  Turing machine model (like $\Class L$ or $\Class P$). We define
  $\Para[\exists^{\leftrightarrow}]C$ as the class of parameterized
  problems $(Q, \kappa)$ for which there exists a $C$-machine $M$, an
  alphabet $\Pi$, and a computable function $\pi\colon\mathbb
  N\to\Pi^*$ such that: For every $x \in \Sigma^*$ we have $x \in Q$
  if, and only if, there exists a bitstring $b \in \{0,1\}^*$ such
  that $M$ accepts with $(x,\pi(\kappa(x)))$ on its input tape and
  $b$ on the two-way choice tape. 
  We define $\Para[\exists^{\to}]C$ similarly, only access to
  the choice tape is now one-way.

  We define $\Para[\exists^{\leftrightarrow}_{\mathit f\!\log}]C$ and
  $\Para[\exists^{\to}_{\mathit f\!\log}]C$ in the same way, but the
  length of $b$ may be at most $|\pi(\kappa(x))| \cdot O(\log n)$.
\end{definition}

\begin{figure}[htpb]
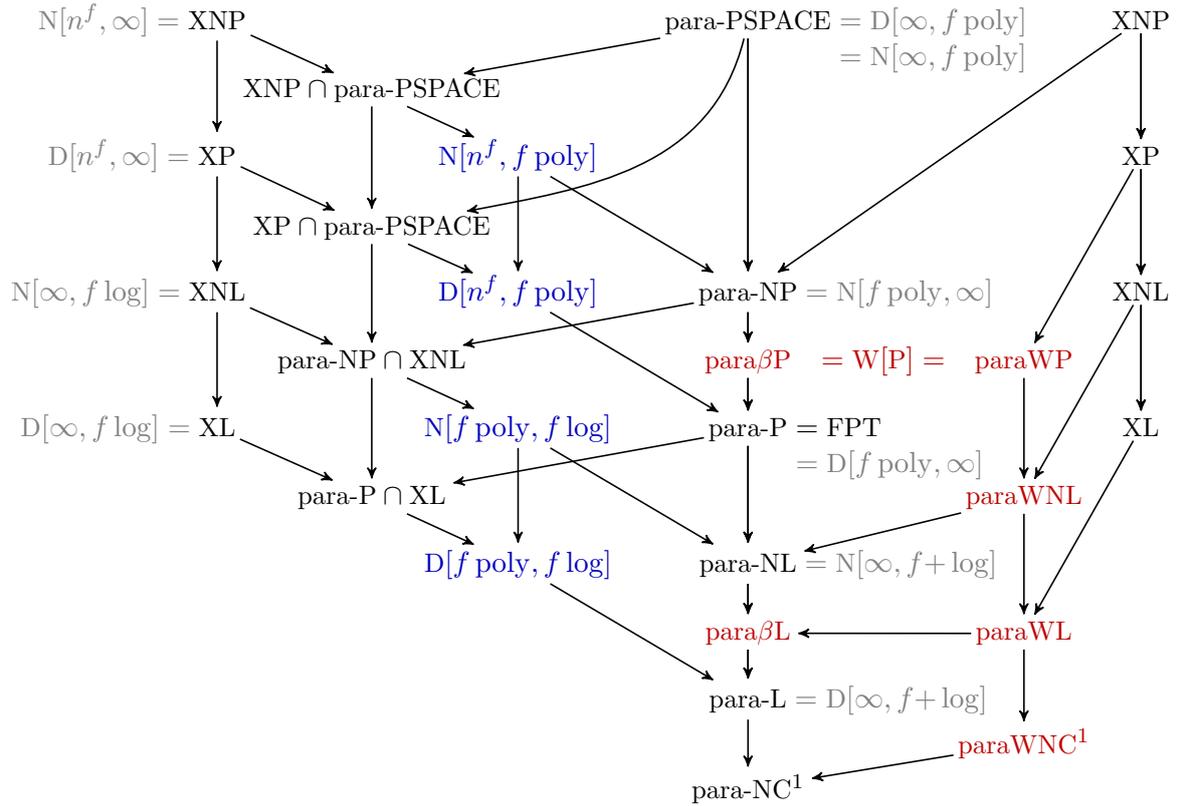

  \def\extra#1{\textcolor{black!50}{#1}}
  \tikz [baseline,xscale=1.1,yscale=1.2,trim left=-10cm,
         red/.style={text=red!75!black},
         blue/.style={text=blue!75!black}] {
    \graph [graph, no placement, nodes={inner sep=1.5pt,anchor=mid}] {
      paraNC1  / {$\Para\Class{NC}^1$} [x=0, y=-1];
      
      paraL  / {$\Para\Class L$\extra{\rlap{${}=\Class D[\infty,f{+}\log]$}}}       
        [x=0, y=0];
        
      paraNL / {$\Para\Class{NL}$\extra{\rlap{${}=\Class  N[\infty,f{+}\log]$}}}
        [x=0, y=1.5];
        
      paraP / {$\Para \Class P$\rlap{\smash{\vtop{${}=\Class{FPT}$\\
                      \extra{${}=\Class D[f\operatorname{poly},\infty]$}}}}}
        [x=0, y=3];
        
      paraNP / {$\Para\Class{NP}$\extra{\rlap{${}=\Class N[f\operatorname{poly},\infty]$}}}
        [x=0, y=4.5];
        
      paraPSPACE / {$\Para\Class{PSPACE}$\extra{\rlap{\smash{%
            \vtop{${}=\Class D[\infty,f\operatorname{poly}]$\\
              ${}=\Class N[\infty,f\operatorname{poly}]$}}}}}
        [x=0, y=7.5];
    
      parabetaL / {$\Para[\beta]\Class L$}              [x=0, y=0.75,red];
      parabetaP / {$\Para[\beta]\Class P$}              [x=0, y=3.75,red];
    
      paraWNC1 / {$\Para[W]\Class{NC}^1$}               [x=3.3, y=-0.5,red];
      paraWL / {$\Para[W]\Class L$}                     [x=3.3, y=0.75,red];
      paraWNL / {$\Para[W]\Class{NL}$}                  [x=3.3, y=2.25,red];
      paraWP / {$\Para[W]\Class P$}                     [x=3.3, y=3.75,red];
    
      XL / {$\Class{XL}$}                               [x=4.7, y=3];
      XNL / {$\Class{XNL}$}                             [x=4.7, y=4.5];
      XP / {$\Class{XP}$}                               [x=4.7,y=6];
      XNP / {$\Class{XNP}$}                             [x=4.7,y=7.5];

      paraNC1 <- paraL <- parabetaL <- paraNL <- paraP <- parabetaP <- paraNP <- paraPSPACE;
      XL <- XNL <- XP <- XNP;
      
      parabetaL <- paraWL;
      parabetaP --[draw=none,"${}=\Class W[\Class P]={}$" red] paraWP;
      
      paraNL <- paraWNL;
      paraNC1 <- paraWNC1 <- paraWL <- paraWNL <- paraWP;
      { paraWL, paraWNL, paraWP} <- {XL, XNL, XP};
      XNP -> paraNP;

      XL2 / {\extra{\llap{$\Class{D}[\infty,f\log]={}$}}$\Class{XL}$}
        [x=-6.35, y=3];
      XNL2 / {\extra{\llap{$\Class{N}[\infty,f\log]={}$}}$\Class{XNL}$}   
        [x=-6.35, y=4.5];
      XP2 / {\extra{\llap{$\Class{D}[n^f,\infty]={}$}}$\Class{XP}$}
        [x=-6.35,y=6];
      XNP2 / {\extra{\llap{$\Class{N}[n^f,\infty]={}$}}$\Class{XNP}$}
        [x=-6.35,y=7.5];

      paraL <- parabetaL <- paraNL <- paraP <- parabetaP <- paraNP <- paraPSPACE;
      XL <- XNL <- XP <- XNP;
      
      paraPiXL / {$\Para\Class P \cap\Class{XL}$}      [x=-4.5, y=2.25];
      paraNPiXNL / {$\Para\Class{NP} \cap\Class{XNL}$} [x=-4.5, y=3.75];
      paraXPiPSPACE / {$\Class{XP} \cap\Para\Class{PSPACE} $}      [x=-4.5, y=5.23];
      paraXNPiPSPACE / {$\Class{XNP} \cap\Para\Class{PSPACE} $}      [x=-4.5, y=6.75];
      
      ftpxl / {$\Class D[f \operatorname{poly},f\log]$} [x=-2.75, y=1.5,blue];
      nftpxl / {$\Class N[f \operatorname{poly},f\log]$} [x=-2.75, y=3,blue];
      
      fpspacexp / {$\Class D[\smash{n^f},f\operatorname{poly}]$} [x=-2.75, y=4.5,blue];
      nfpspacexp / {$\Class N[\smash{n^f},f\operatorname{poly}]$} [x=-2.75, y=6,blue];

      paraXNPiPSPACE -> paraXPiPSPACE -> paraNPiXNL -> paraPiXL; 
      
      nfpspacexp -> fpspacexp, nftpxl -> ftpxl;
      
      {XL2, paraP} -> paraPiXL -> ftpxl -> paraL;
      {XNL2, paraNP} -> paraNPiXNL -> nftpxl -> paraNL;
      
      {XP2} -> paraXPiPSPACE -> fpspacexp -> paraP;
      {XNP2, paraPSPACE} -> paraXNPiPSPACE -> nfpspacexp -> paraNP;
      
      paraPSPACE ->[out=-105,in=10] paraXPiPSPACE;
      
      XL2 <- XNL2 <- XP2 <- XNP2;      
    };
    \path (-9,0);
    }
  \caption{In this inclusion diagram bounded nondeterminism classes
    are shown in red and time--space classes in blue. The X-classes
    are shown twice to keep the diagram readable. All known inclusions
    are indicated, where $C \to D$ means $C \supseteq D$.}
  \label{fig-inclusions}
\end{figure}

Observe that, as argued earlier,
$\Para[\exists^{\leftrightarrow}]\Class L =
\Para[\exists^{\leftrightarrow}]\Class P = \Para[\exists^{\to}] \Class
P = \Para\Class{NP}$ and $\Para[\exists^{\to}]\Class L =
\Para\Class{NL}$.  Also observe that 
$\Para[\exists^{\leftrightarrow}_{\mathit f log}] \Class P =
\Para[\exists^{\to}_{\mathit f log}] \Class P = \Class W [ \Class P ]$
by one of the many possible definitions of~$\Class W[\Class P]$.

The above definition can easily be extended to the case where a
universal quantifier is used instead of an existential one and where
\emph{sequences} of quantifiers are used. This is interpreted in the
usual way as having a choice tape for each quantifier and the
different ``exists \dots\,for all''-conditions must be met in the
order the quantifiers appear. For instance, for problems in
$\Para[\exists^{\leftrightarrow}_{\mathit
  f\!\log}\exists^{\to}]\Class L$ we have $x 
\in Q$ if, and only if, there exists a bitstring of length
$f_x\log_2 n$ for the first, two-way-readable choice
tape for which an $\Class{NL}$-machine accepts.
The classes $\Class{para\text-NL}[f\log]$,
$\Class{para\text-L\text-cert}$, and $\Class{para\text-NL\text-cert}$
introduced in an ad hoc manner by Elberfeld et al.\ in
\cite{ElberfeldST2012} can now be represented systematically: They are 
$\Para[\exists^{\to}_{\mathit f\!\log}]\Class L$,
$\Para[\exists^{\leftrightarrow}_{\mathit f\!\log}]\Class L$, and  
$\Para[\exists^{\leftrightarrow}_{\mathit f\!\log}\exists^{\to}]\Class
L$, respectively.

In order to make the notation more useful in practice, instead of
``$\exists^{\to}$'' let us write ``$\Class{N}$'' and instead of
``$\exists^{\to}_{\mathit   f\!\log}$'' we write ``$\beta$'' as is
customary. As a new notation, instead of
``$\exists^{\leftrightarrow}_{\mathit  f\!\log}$'' and
``$\forall^{\leftrightarrow}_{\mathit  f\!\log}$''  we write
``$\Class{W}$'' and ``$\Class{W_\forall}$,'' respectively. The three
classes of \cite{ElberfeldST2012} now become $\Para[\beta] 
\Class L$, $\Para[W] \Class L$, and $\Para[W N] \Class L$.

Our reasons for using ``W'' to denote
$\exists^{\leftrightarrow}_{\mathit f\!\log}$ will be explained fully in
Section~\ref{section-justification-w}; for the moment just observe
that $\Class W[\Class P] = \Para[W] \Class P$ holds.
To get a better intuition on the W-operator, note that it provides
machines with ``$f_x \log_2 n$ bits of nondeterministic information''
or, equivalently, with ``$f_x$ many nondeterministic positions in the
input'' and these bits are provided as part of the input. This allows
us to also apply the W-operator to classes like $\Class{NC}^1$ that
are not defined in terms of Turing machines.

The right half of Figure~\ref{fig-inclusions} depicts the known
inclusions between the introduced classes, the left half shows the
classes introduced next.

\subsection{Parameterized Time--Space Classes}
\label{section-intro-ftpxp}

In classical complexity theory, the major complexity classes are
either defined in terms of time complexity ($\Class P$, $\Class{NP}$,
$\Class{EXP}$) or in terms of space complexity ($\Class L$, $\Class{NL}$,
$\Class{PSPACE}$), but not both at the same time: by the
well-known inclusion chain $\Class L \subseteq
\Class{NL} \subseteq \Class P \subseteq \Class{NP} \subseteq
\Class{PSPACE} = \Class{NPSPACE} \subseteq \Class {EXP}$ 
space and time are intertwined in such a way that bounding either 
automatically bounds the other in a specific way (at least for the
major complexity classes).
In the parameterized world, interesting new classes arise when we restrict
time and space simultaneously: namely whenever the time is
``para-restricted'' while space is ``X-restricted'' or \emph{vice versa.}

\begin{definition}
  For a space bound $s$ and a time bound $t$, both of which may depend
  on a parameter $k$ and the input length $n$, let $\Class D[t,s]$ 
  denote the class of all parameterized problems that can be accepted
  by a deterministic Turing machine in time $t(\kappa(x),|x|)$ and
  space $s(\kappa(x),|x|)$. Let $\Class N[t,s]$ denote the
  corresponding nondeterministic class.
\end{definition}

Four cases are of interest: First, $\Class{D}[f \mathrm{poly},
f\log]$, meaning that $t(k,n) = f(k) \cdot n^{O(1)}$ and $s(k,n) =
f(k) \cdot O(\log n)$, contains all problems that are ``fixed
parameter tractable via a machine needing only slice-wise logarithmic
space,'' and, second, the nondeterministic  counterpart $\Class N[f
\mathrm{poly}, f\log]$. The two other cases are $\Class D[n^f,
f\mathrm{poly}]$ and 
$\Class N[n^f, f\mathrm{poly}]$, which contain problems that are ``in
slice-wise polynomial time via machines that need only fixed
parameter polynomial space.'' See Figure~\ref{fig-inclusions} for the
trivial inclusions between the classes.

In Section~\ref{section-lcs} we will see that these classes are not only of
scholarly interest. Rather, we will show that $\Lang{lcs}$
parameterized by the number of input strings is complete for $\Class
N[f \mathrm{poly}, f\log]$.

\section{Complete Problems for Bounded Nondeterminism}
\label{section-bounded-non}
\label{section-justification-w}
    
In this section we present new natural problems that are
complete for $\Para[W]\Class{NL}$ and $\Para[W]\Class{L}$. Previously,
it was only known that the following ``colored reachability
problem''~\cite{ElberfeldST2012} is complete for $\Para[W]\Class{NL}$: We
are given an edge-colored graph, two vertices $s$ and $t$, and a
parameter~$k$. The question is whether there is a path from $s$ to $t$
that uses only $k$ colors. Our key tool for proving new completeness
results will be the introduction of a ``union operation,'' which turns
$\Class P$-, $\Class{NL}$-, and $\Class L$-complete problems into
$\Para[W]\Class P$-, $\Para[W]\Class{NL}$-, and $\Para[W]\Class
L$-complete problems, respectively. Building on this, we prove the
parameterized associative generability problem to be complete for
$\Para[W]\Class{NL}$. Note that the underlying classical problem is
well-known to be $\Class{NL}$-complete and, furthermore, if we drop
the requirement of associativity, the parameterized and
classical versions are known to be complete for
$\Para[W]\Class P$ and~$\Class P$, respectively. 

At this point, we remark that Guillemot, in a
paper~\cite{Guillemot2011} on parameterized \emph{time} complexity,
uses ``$\Class{WNL}$'' to denote a class different from the class
$\Para[W]\Class{NL}$ defined in this paper. Guillemot chose the name
because his definition of the class is derived from one possible 
definition of $\Class{W}[1]$ by replacing a time by a space
constraint. Nevertheless, we believe that our definition of a
``W-operator'' yields the ``right analogue'' of $\Class W[\Class P]$:
First, there is the above pattern that parameterized version of
problems complete  $\Class P$, $\Class{NL}$, and $\Class L$ tend to be
complete for $\Para[W]\Class P$, $\Para[W]\Class{NL}$, and
$\Para[W]\Class L$, respectively. Furthermore, in
Section~\ref{section-lcs} we show that the class  
$\Class{WNL}$ defined and studied by Guillemot is exactly the
fpt-reduction closure of the time--space class $\Class
N[f\operatorname{poly},f\log]$.

\paragraph{Union Problems.}
For numerous problems studied in complexity theory the input consists
of a string in which some positions can be ``selected'' and the
objective is to select a ``good'' subset of these positions. For
instance, for the satisfiability problem we must select some
variables such that setting them to true makes a formula
true; for the circuit satisfiability problem we must select
some input gates such that when they are set to~$1$ the circuit
evaluates to~$1$; and for the exact cover problem we 
must select some sets from a family of sets so that they form a
partition of the union of the family. In the following, we introduce some
terminology that allows us to formulate all of these problems in a
uniform way and to link them to the W-operator.

Let $\Sigma$ be an alphabet that contains none of the three special
symbols $?$, $0$, and~$1$. We call a word $t \in (\Sigma \cup \{?\})^*$ a
\emph{template}. We call a word $s \in (\Sigma \cup \{0,1\})^*$ an
\emph{instantiation of $t$} if $s$ is obtained from $t$ by replacing
exactly the $?$-symbols arbitrarily by $0$- or $1$-symbols. Given
instantiations $s_1, \dots, s_k$ of the same template~$t$, their
\emph{union} $s$ is the instantiation of $t$ that has a $1$ exactly at
those positions $i$ where at least one $s_j$ has a $1$ at position~$i$
(the union is the ``bitwise or'' of the instantiated positions and is
otherwise equal to the template). 

Given a language $A \subseteq (\Sigma\cup\{0,1\})^*$, we define three
different kinds of union problems for~$A$. Each of them is a
parameterized problem where the parameter is~$k$. As we will see in a
moment, the first kind is linked to the W-operator 
while the last kind links several well-known languages from classical 
complexity theory to well-known parameterized problems. We will also
see that the three kinds of union problems for a language~$A$
often all have the same complexity. 
\begin{enumerate}
\item The input for $\PLang{family-union-}A$ are a template $t \in
  (\Sigma \cup \{?\})^*$ and a family $(S_1, \dots, S_k)$ of $k$ sets
  of instantiations of~$t$. The question is whether there are $s_i \in
  S_i$ for $i \in \{1,\dots,k\}$ such that the union of
  $s_1,\dots,s_k$ lies in $A$.
\item The input for $\PLang{subset-union-}A$ are a template $t \in
  (\Sigma \cup \{?\})^*$, a set $S$ of instantiations of~$t$, and a
  number~$k$. The question is whether there exists a subset $R
  \subseteq S$ of size $|R| = k$ such that the union of $R$'s elements
  lies in $A$.
\item The input for $\PLang{weighted-union-}A$ are a template $t \in
  (\Sigma \cup \{?\})^*$ and a number~$k$. The question is whether
  there exists an instantiation $s$ of $t$ containing exactly 
    $k$ many $1$-symbols such that $s \in A$?
\end{enumerate}

To get an intuition for these definitions, think of instantiations as
words written on transparencies with $0$ rendered as an empty box and
$1$ as a checked box. Then for the family union problem we are given
$k$ heaps of transparencies and the task is to pick one transparency
from each heap such that ``stacking them on top of each other'' yields an 
element of~$A$. For the subset union problem, we are only given one
stack and must pick $k$ elements from it. We call the weighted union
problem a ``union'' problem partly in order to avoid a clash with
existing terminology and partly because the weighted union problem is
the same as the subset union problem for the special set $S$ containing
all instantiations of the template of weight~$1$.

Concerning the promised link between well-known languages and
parameterized problems, consider $A= \Lang{circuit-value-problem}$
($\Lang{cvp}$) where we use $\Sigma$ to encode a circuit and use $0$'s
and $1$'s solely to describe an assignment to the input
gates. Then the input for \PLangText{weighted-union-cvp} are a circuit
with $?$-symbols instead of a concrete assignment together with a
number $k$, and the question is whether we can replace exactly $k$
of the  $?$-symbols by $1$'s (and the other by 
$0$'s) so that the resulting instantiation lies in $\Lang{cvp}$. Clearly,
$\PLang{weighted-union-cvp}$ is exactly the $\Class W[\Class
P]$-complete problem $\PLang{circuit-sat}$, which asks whether there is a
satisfying assignment for a given circuit that sets exactly $k$ input
gates to~$1$. 

Concerning the promised link between the union problems and the
W-op\-er\-a\-tor, recall that the operator provides machines with $f_x$
nondeterministic indices as part of the input. In particular, a
W-machine can mark $f_x$ different ``parts'' of the input -- like one
element from each of $f_x$ many sets in a family, like the elements of a
size-$f_x$ subset of some set, or like $f_x$ many positions in a
template. With this observation it is not difficult to see that if $A
\in C$, then all union versions of $A$ lie in $\Para[W]C$. 
A much deeper observation is that the union versions are also often
\emph{complete} for these classes. In the next theorem, which states this
claim precisely, the following definition of a
\emph{compatible logspace projection $p$} from a language $A$ to a
language~$B$ is used: First, $p$ must be a logspace reduction from
$A$ to~$B$. Second, $p$ is a projection, meaning that each symbol of
$p(x)$ depends on at most one symbol of~$x$. Third, for each word
length $n$ there is a single template $t_n$ such for all $x\in\Sigma^n$
the word $p(x)$ is an instantiation of~$t_n$.

\begin{theorem}\label{theorem-comp}
  Let $C \in \{\Class{NC}^1, \Class L, \Class{NL}, \Class P\}$.
  Let $A$ be complete for~$C$ via compatible logspace
  projections. Then $\PLang{family-union-}A$ is complete for 
  $\Para[W]C$ under para-$\Class L$-reductions.\footnote{The proof
    shows that the theorem actually also holds for any ``reasonable''
    class $C$ and any ``reasonable'' weaker reduction.} 
\end{theorem}

\begin{proof}
  For containment, on input of a template $t$ and family
  $(S_1,\dots,\penalty0 S_k)$ of sets of instantiations of~$t$, a
  $\Para[W]C$-\penalty0machine or -circuit interprets its $k \log_2 n$ 
  nondeterministic bits as $k$ indices, one for each $S_i$. Let $s_i
  \in S_i$ be the elements selected in this way. We run a simulation of
  the $C$-machine (or $C$-circuit) that decides $A$ on the union $s$
  of $s_1,\dots,s_k$. For logspace machines, we may not have enough 
  space to write $s$ on a tape, so whenever the machine would like to
  know the $j$th bit of $u$, we simply compute the bitwise-or of the
  $j$th positions of the $s_i$. 
  
  For hardness, consider any
  problem $(Q,\kappa) \in \Para[W]C$. By definition, this means the following:
  There are a language $X \subseteq \Gamma^*$ in~$C$ and computable
  functions $\pi \colon 
  \mathbb N \to \Pi^*$ and $f \colon \mathbb N \to \mathbb N$ such 
  that $x \in Q$ if, and only if, there is a string $b \in \{0,1\}^{f_x \log_2 n}$
  with $(x,\pi(\kappa(x)),b) \in X$. Furthermore, since 
  $A$ is complete for~$C$ via compatible logspace projections, we can
  reduce $X$ to $A$ via some~$p$. (As always,  $n = |x|$ and $f_x = f(\kappa(x))$.)

  For the reduction of $(Q,\kappa)$ to $\PLang{family-union-}A$,
  let an input $x$ be given. Our para-$\Class L$-re\-duc\-tion
  first computes $\pi(\kappa(x))$. Since our reduction $p$ is compatible,
  for all possible $b$ the string $p(x,\pi(\kappa(x)),b)$ will have
  $0$-symbols and $1$-symbols at the same positions and all other
  positions will not vary with~$b$ at all. Our template~$t$ will
  be the string $p(x,\pi(\kappa(x)),b)$ with $?$-symbols placed at
  these positions (as argued, we can use any $b$).

  To define the sets of instances $S_i$, observe that the strings $b \in 
  \{0,1\}^{f_x \log_2 n}$ can be thought of as sequences of $f_x$
  symbols from the alphabet $\Delta = \{0,1\}^{\log_2 n}$, whose
  elements we call \emph{blocks}. For $i \in \{1,\dots,f_x\}$  let $S_i =
  \{m_i^\delta \mid \delta \in \Delta\}$ where $m_i^\delta$ replaces
  the $r$th $?$-symbol of the template~$t$ by $c \in \{0,1\}$ as
  follows: If the $r$th position depends on a symbol in
  $(x,\pi(\kappa(x)),b)$ that lies in a block of $b$, but not in the
  $i$th block, let $c = 0$. Otherwise, let $c$ be whatever symbol ($0$
  or $1$) the reduction outputs when the $i$th block is set
  to~$\delta$. This concludes the construction. 

  As an example for the construction, suppose the reduction $p$ simply
  doubles its input (so $w$ is mapped to $w w$) and $\pi$ just returns the
  empty string, and $\Sigma = \{\alpha,\beta,\gamma\}$. Consider, say, $x =
  \alpha\beta\gamma\alpha$ and assume $f_x=2$. We then 
  have $\Delta = \{00,01,10,11\}$. The reduction would produce two
  sets $S_1$ and $S_2$. For $S_1$, we have a look at what $p$ does on
  input of a string like $(x,\pi(\kappa(x)),b)$. For simplicity let us
  ignore parentheses and commas, and consider $b=1111$, so this string
  would just be $\alpha\beta\gamma\alpha1111$. The reduction maps this to
  $\alpha\beta\gamma\alpha1111\alpha\beta\gamma\alpha1111$. In this string, the fifth, sixth, thirteenth, 
  and fourteenth bits actually depend on the first block of
  $\alpha\beta\gamma\alpha1111$, so the reduction would produce the first set
  \begin{align*}
    S_1 =
    \{
    &\alpha\beta\gamma\alpha0000\alpha\beta\gamma\alpha0000,\\
    &\alpha\beta\gamma\alpha0100\alpha\beta\gamma\alpha0100,\\
    &\alpha\beta\gamma\alpha1000\alpha\beta\gamma\alpha1000,\\
    &\alpha\beta\gamma\alpha1100\alpha\beta\gamma\alpha1100\}. 
  \end{align*}
  In a similar manner, the
  reduction would produce the second set
  \begin{align*}
    S_2 = \{
    &\alpha\beta\gamma\alpha0000\alpha\beta\gamma\alpha0000,\\
    &\alpha\beta\gamma\alpha0001\alpha\beta\gamma\alpha0001,\\
    &\alpha\beta\gamma\alpha0010\alpha\beta\gamma\alpha0010,\\
    &\alpha\beta\gamma\alpha0011\alpha\beta\gamma\alpha0011\}.
  \end{align*}
  Observe
  that, indeed, we can get every string $p(x,\pi(\kappa(x)),b)$ by
  taking the union of one string from $S_1$ and one string
  from~$S_2$. 

  To see that the reduction is correct, consider the union of the
  elements of any set $\smash{\{m_1^{\delta_1}},\dots,\penalty
  0m_{f_x}^{\delta_{f_x}}\}$ where the $m_i^{\delta_i}$ are chosen
  from the different $S_i$. By construction, their union will be
  exactly the image of $(x,\pi(\kappa(x)),\delta_1\dots\delta_{f_x})$
  under~$p$. In particular, $x \in Q$ holds if, and only if, we can choose one
  instantiation from each $S_i$ such that their union is in~$A$.
\end{proof}

\paragraph{Parameterized Satisfiability Problems.}

Recall that the problem \PLangText{weighted-union-cvp} equals
\PLangText{circuit-sat}. 
Since one can reduce \PLangText{family-union-cvp} to
\PLangText{weighted-union-cvp} (via essentially the same reduction as
that used in the proof of Theorem~\ref{theorem-nc1} below),
Theorem~\ref{theorem-comp} provides us with a direct proof that
\PLangText{circuit-sat}${}={}$\PLangText{weighted-union-cvp} is
complete for $\Para[W]\Class P$.  We get an even more interesting
result when we apply the theorem to $\Lang{bf}$, the propositional
formula evaluation problem. We encode pairs of formulas and
assignments in the straightforward way by using $0$ and $1$ solely for the
assignment. Since $\Lang{bf}$ is complete for $\Class{NC}^1$ under
compatible logspace reductions, see~\cite{Buss1987,Bussetal1992},
\PLangText{family-union-bf} is complete for $\Para[W]\Class{NC}^1$ by
Theorem~\ref{theorem-comp}. By further reducing the problem to
\PLangText{weighted-union-bf}, we obtain: 

\begin{theorem}\label{theorem-nc1}
  $\PLang{weighted-union-bf}$ is para-$\Class L$-complete\footnote{As
    in Theorem~\ref{theorem-comp} one can also use weaker reductions.}
  for $\Para[W]\Class{NC}^1$.
\end{theorem}

\begin{proof}
  The language $\Lang{bf}$ is complete for $\Class{NC}^1$,
  see~\cite{Buss1987,Bussetal1992}, and completeness can be achieved
  by compatible projections: Indeed, for input words of the same
  length, the reduction will map them to the same formula, only the
  assignment to the variables will differ (the input word is encoded
  solely in this assignment). Thus, by Theorem~\ref{theorem-comp} we
  get that $\PLang{family-union-bf}$ is complete for
  $\Para[W]\Class{NC}^1$ under para-$\Class L$-reductions (actually,
  also under weaker reductions like parameterized first-order
  reduction, but they are not in the focus of this paper).
  
  We now show that $\PLang{family-union-bf}$ reduces to
  $\PLang{subset-union-bf}$, which in turn reduces to
  $\PLang{weighted-union-bf}$. For the first reduction, let the sets $S_1$
  to $S_k$ be given as input. All elements $s_{ij}$ of the $S_i$
  represent assignments to the variables of the same
  formula~$\phi$. Our aim is to construct a set $S$ and a new formula
  $\phi' = \phi \land \psi$, where the job of $\psi$ is to ensure
  that any selection of $k$ elements from $S$ can only lead to $\phi'$
  being true if the selection corresponds to picking ``exactly one
  element from each $S_i$.'' In detail, for each $s_{ij}$ we introduce a
  new variable $v_{ij}$. The assignment $s'_{ij}$ for $\phi'$ is the
  same as $s_{ij}$ for the ``old'' variables and is $1$ only for
  $\smash{v_{ij}}$ among the new variables ($\smash{v_{ij}}$ ``tags''
  $\smash{s_{ij}}$). As an example, suppose there are three variable
  $x$, $y$, and $z$ in $\phi$ and suppose $S_1 = \{\phi000, \phi001\}$
  (meaning that one assignment sets all variables to false and the
  other sets only $z$ to true) and $S_2 = \{\phi001\}$. Then there
  would be three additional new variables and $S 
  = \{\phi'000\,100, \phi'001\,010, \phi'001\,001\}$. 
  Now, setting $\psi = \bigwedge_{i=1}^k \bigvee_{j=1}^{|I_i|} v_{ij}$
  ensures that $\psi$ will only be true for the union of $k$
  assignments taken from $S$ if exactly one assignment was taken from 
  each~$S_i$. 

  Next, we reduce $\PLang{subset-union-bf}$ to
  $\PLang{weighted-union-bf}$. Towards this end, let $S =
  \{b_1,\dots,b_n\}$ be given as input and let $\phi$ be the formula
  underlying the~$s_i$. Our new formula $\phi'$ has exactly
  $n$ variables $v_1$ to~$v_n$ and is obtained from $\phi$ by leaving
  the structure of $\phi$ identical, but substituting each occurrence of
  a variable $x$ as follows: let $X \subseteq \{1,\dots,n\}$ be the
  set of indices $i$ such that in $s_i$ the variable $x$ is set
  to~$1$. Then we substitute $x$ by $\bigvee_{i\in X} v_i$. The output
  $S'$ of the reduction is $\phi'$ together with all assignments
  making exactly one of the variables $v_i$ true. As an example, let
  $\phi = x \land (y \to x) \land z$ and let $S = \{\phi000, \phi101, \phi 010,
  \phi111\}$. Then there would be four variables $v_1$ to $v_4$ and
  the formula $\phi'$ would be $v_2 \land ((v_3 \lor v_4) \to v_2)
  \land (v_2 \lor v_4)$ and the set $S'$ would be $\{\phi'0001,
  \phi'0010, \phi'0100, \phi'1000\}$. 

  To see that this reduction is correct, first assume that we have
  $(t,S,k) \in   \PLang{subset-union-bf}$ via a selection $\{s_1,\dots,s_k \}\subseteq
  S$. Then we also have $(t',S',k) \in \PLang{weighted-union-bf}$ via the $k$ elements of
  $S'$ where exactly the variables corresponding to $s_1$ to $s_k$ are
  set to true: in $\phi'$ the expressions  $\bigvee_{i\in X} v_i$ that
  was substituted for a variable $x$ will be true exactly if one of
  the $s_i$ has set $x$ to~$1$. Thus, 
  $\phi'$ will evaluate to $1$ for the assignment in which exactly the
  selected $v_i$ are true if, and only if, $\phi$ evaluates to true
  for the ``bitwise or'' of the assignments $s_1,\dots,s_k$ -- which
  it does by assumption. For the other direction, assume that $(t',S',k)
  \in \PLang{weighted-union-bf}$. Then, by essentially the same argument,
  we obtain a subset of $S$ whose bitwise or makes $\phi$ evaluate
  to~$1$.
\end{proof}

By definition, $\Class W[\Lang{sat}]$ is the fpt-reduction closure of
$\PLang{weighted-sat}$, which is the same as \PLangText{weighted-union-bf}. Thus, by the theorem,
$\Class W[\Lang{sat}]$ is also the fpt-reduction closure of
$\Para[W]\Class{NC}^1$ --~a result that may be of independent interest. For example, it 
shows that $\Class{NC}^1 = \Class P$ implies $\Class W[\Lang{sat}] =
\Class W[\Class P]$. Note that we do not claim $\Class
W[\Lang{sat}]=\Para[W]\Class{NC}^1$ since $\Para[W]\Class{NC}^1$ is
presumably not closed under fpt-reductions.

\paragraph{Graph Problems.}

In order to apply Theorem~\ref{theorem-comp} to standard graph
problems like $\Lang{reach}$ or $\Lang{cycle}$, we encode graphs using
adjacency matrices consisting of $0$- and $1$-symbols. Then a template
is always a string of $n^2$ many $?$-symbols for $n$ vertex
graphs. The ``colored reachability problem''
mentioned at the beginning of this section equals
$\PLang{subset-union-reach}$.\footnote{For exact equality, in the
  colored reachability problem we must allow edges to have several
  colors, but this is does not change the problem complexity.}  Note
that any reduction to a union problem for this encoding is
automatically compatible as long as the number of vertices in the
reduction's output depends only on the length of its input.   

Applying Theorem~\ref{theorem-comp} to standard $\Class{L}$- or
$\Class{NL}$-complete problems yields that their family union versions are complete for
$\Para[W]\Class L$ and $\Para[W]\Class{NL}$, respectively. By reducing
the family versions further to the subset union version, we get
the following: 

\begin{theorem}\label{theorem-unions}
  For $A \in \{\Lang{reach}, \Lang{dag-reach}, \Lang{cycle}\}$,
  $\PLang{subset-union-}A$ is complete for 
  $\Para[W]\Class{NL}$, while for $B \in \{\Lang{undirected-reach},$
  $\Lang{tree},$ $\Lang{forest},$ $\Lang{undirected-cycle} \}$, the
  problem $\PLang{subset-union-}B$ is complete for $\Para[W]\Class L$.
\end{theorem}

\begin{proof}
  For each of the problems, we reduce its family union version to
  it. This suffices: By Theorem~\ref{theorem-comp} and the fact
  that the underlying problems like $\Lang{reach}$ are complete for
  $\Class{NL}$ and $\Class L$ under compatible logspace projections
  (even under first-order projections), the family versions are
  complete for the respective classes.
  
  Recall that the difference between the problems $\PLang{family-union-}A$ and
  $\PLang{subset-union-}A$ is 
  that in the first we are given $k$ sets $S_i$ from each of which we
  must choose one element, while for the latter we can pick $k$
  elements from a single set $S$ arbitrarily. If the reduction were to
  just set $S$ to the union of the $S_i$, then many choices of $k$
  sets of $S$ will correspond to taking multiple elements from a
  single~$S_i$. In such cases, their union should \emph{not} be an
  element of~$A$. 

  \begin{figure}[htpb]
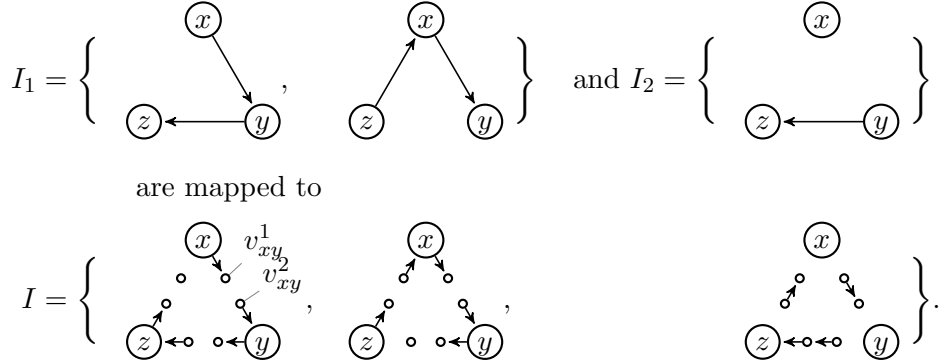

    \tikzset{
      p1/.style = {graphs/math nodes,at={(90:9mm)}},
      p2/.style = {graphs/math nodes,at={(-30:9mm)}},
      p3/.style = {graphs/math nodes,at={(-150:9mm)}},
      v/.style= {graphs/as=,minimum size=3pt, inner sep=0pt,text height=,text depth=,
},
      xy1/.style= {v,at={($ (90:9mm)!.375!(-30:9mm) $)}},
      xy2/.style= {v,at={($ (90:9mm)!.625!(-30:9mm) $)}},
      yz1/.style= {v,at={($ (-30:9mm)!.375!(-150:9mm) $)}},
      yz2/.style= {v,at={($ (-30:9mm)!.625!(-150:9mm) $)}},
      zx1/.style= {v,at={($ (-150:9mm)!.375!(90:9mm) $)}},
      zx2/.style= {v,at={($ (-150:9mm)!.625!(90:9mm) $)}},
      my pic/.style = {baseline,
        graphs/every graph/.style={graph, nodes={node,
            inner sep=1pt, minimum size=2mm}, no placement}
      }
    }
    \begin{align*}
      \begin{array}{rllrl}
        I_1 = \Biggl\{&
        \tikz[my pic] \graph { x[p1] -> y[p2] -> z[p3] };, &
        \tikz[my pic] \graph { x[p1] -> y[p2],  z[p3] -> x };\ 
        \Biggr\}& \text{ and }
        I_2 = \Biggl\{&
        \tikz[my pic] \graph { x[p1],  y[p2] -> z[p3] };\ 
        \Biggr\} \\[2.5em]
        &\text{ are mapped to}\\[.5em]
        I = \Biggl\{&
        \tikz[my pic] {\graph {
          vxy1[xy1], vxy2[xy2], vyz1[yz1], vyz2[yz2], vzx1[zx1], vzx2[zx2],
          x[p1], y[p2], z[p3],
          x -> vxy1, vxy2 -> y -> vyz1, vyz2 -> z -> vzx1 };
        \node [above right=2mm, inner sep=1pt,] (l1) at (vxy1) {$v_{xy}^1$};
        \node [above right=2.5mm, yshift=-1.5mm,inner sep=1pt,] (l2) at (vxy2) {$v_{xy}^2$};
        \draw [help lines] (l1) -- (vxy1);
        \draw [help lines] (l2) -- (vxy2);
        }, &
        \tikz[my pic] \graph {
          vxy1[xy1], vxy2[xy2], vyz1[yz1], vyz2[yz2], vzx1[zx1], vzx2[zx2],
          x[p1], y[p2], z[p3],
          x -> vxy1, vxy2 -> y -> vyz1, z -> vzx1, vzx2 -> x };
        ,&&
        \tikz[my pic] \graph {
          vxy1[xy1], vxy2[xy2], vyz1[yz1], vyz2[yz2], vzx1[zx1], vzx2[zx2],
          x[p1], y[p2], z[p3],
          x, vxy1 -> vxy2,  y, vyz1 -> vyz2 -> z, vzx1 -> vzx2 };\ 
        \Biggr\}.        
      \end{array}
    \end{align*}
    \caption{An example of the reduction from a family
      union graph problem to a subset union graph problem. In the
      example, $V = \{x,y,z\}$. The $s_i$ are indicated as edge sets
      even though, in reality, they are bitstrings encoding adjacency
      matrices. The small vertices are the $v_{ab}^1$ and $v_{ab}^2$,
      but only those for $(a,b) \in \{(x,y), (y,z), (z,x)\}$ are shown.}  
    \label{fig-red}
  \end{figure}
  
  To achieve the effect that the union of a subset of $S$ with
  multiple elements from the same $S_i$ is not in~$A$, we use the same
  construction for all $A$, except for $A = \Lang{forest}$. The
  construction works as follows: Since the $S_i$ are compatible, they
  are defined over the same set~$V$ of vertices. Each $s \in S_i$
  encodes an edge set $E_s \subseteq V^2$. We construct a new vertex
  set $V' \supseteq V$ as follows: For each pair $(a,b) \in V^2$ we
  introduce $k$ new vertices $v_{ab}^1$, \dots, $v_{ab}^k$ and add
  them to~$V'$. For each $s\in S_i$ we define a new edge
  set $E'_s \subseteq V' \times V'$ as follows: First, for each $(a,b)
  \in V^2$ let $(v_{ab}^{i-1},v_{ab}^i) \in E'_s$, where $v_{ab}^0 =
  a$. Second, for each $(a,b) \in E_s$, let $(v_{ab}^k,b)
  \in E'_s$. Let $s'$ be the bitstring encoding the adjacency matrix
  of~$E'_s$. We set $S = \{\,s' \mid s \in S_i \text{ for some
    $i$}\}$. An example of how this reduction works is depicted in
  Figure~\ref{fig-red}.
  
  In order to argue that the reduction works for all problems, we make
  two observations. Given any subset
  $\{s'_1,\dots,s'_k\} \subseteq S$, for each $s'_i$ there is a unique
  corresponding~$s_i$, lying in (some) $S_j$. Let $G' = (V',E')$ denote graph whose
  adjacency matrix is the union of $\{s'_1,\dots,s'_k\}$ and,
  correspondingly, let $G = (V,E)$ be the union of the $S_i$. Now, first
  assume that, indeed, we have $s_i \in S_i$ for all $i \in
  \{1,\dots,k\}$. Then for every pair $(a,b) \in V$ 
  the new vertices $v_{ab}^1$ to $v_{ab}^k$ will form a path in
  $G'$ attached to~$a$. Furthermore, for every edge $(a,b) \in E$
  there is a path from 
  $a$ to $b$ in $E'$. On the other hand, for $(a,b) \notin E$, we
  cannot get from $a$ to $b$ in $G'$ using only new vertices: the edge
  $v_{ab}^k \to b$ will be missing. This
  proves our first observation: for vertices $a,b \in V$ there is a
  path from $a$ to $b$ in $G'$ if, and only if, there is such a path
  in $G$. Our second observation concerns the case that there are two
  strings $s'_i$ and $s'_j$ such that $s_i$ and $s_j$ lie in the same
  set $S_x$. In this case, for every two vertices $a,b \in V$ at least
  one edge is missing along the path $v_{ab}^0$ to
  $v_{ab}^k$. Thus, we can observe that there is no path from any
  $a \in V$ to any other $b \in V$ in $G'$.

  Let us now argue that the reduction is correct: For the reachability
  problem, by the first observation reachability is correctly
  transferred from $G$ to $G'$ and by the second observation no
  ``wrong'' choice of $s'_i$ will induce reachability. The exact same
  argument holds for undirected reachability and reachability in
  \textsc{dag}s. For trees and cycles, the argument also works since trees and
  cycles remain trees and cycles for ``correct'' choices of the $s'_i$
  and they get destroyed for any ``wrong'' choice.

  For forests, the reduction described above does not work since in
  case several $s'_i$ are picked such that their $s_i$ stem from the
  same~$S_j$, the graph $G'$ becomes a collection of small trees: a
  forest -- and this is exactly what should \emph{not} happen. The
  trick is to use a different reduction: For every pair $s_i, s_j \in
  S_x$ for $x \in \{1,\dots,k\}$ we add three new vertices to the
  graph: $a$, $b$, and $c$. In $s'_i$ we add the edges $(a,b)$ and
  $(b,c)$, in $s'_j$ we add the edge~$(c,a)$. Now, clearly, whenever
  $s'_i$ and $s'_j$ are picked stemming from the same $S_x$, a cycle
  will ensue; and if only one $s_i$ is picked from each $S_i$, 
  paths of length $1$ or $2$ will result in the new vertices that do
  not influence whether the graph is a forest or not.
\end{proof}

To conclude this section on union graph problems, we would like to point out
that one can also ask which problems are complete for the
``co-W-classes'' $\Para[W_\forall]\Class{NL}$ and
$\Para[W_\forall]\Class L$. It is straightforward to see that an
analogue of Theorem~\ref{theorem-comp} holds if we define problems 
$\PLang{partitioned-\penalty0union$_\forall$-}A$ as a ``universal
version'' of the partitioned union problem (we ask whether \emph{for
  all} choices of $b_i$ their union is in $A$). For instance,
$\PLang{partitioned-union$_\forall$-cycle}$ is complete for
$\Para[W_\forall]L$. It is also relatively easy to employ the same ideas as
those from the proof of Theorem~\ref{theorem-unions} to show that
the universal union versions of all problems mentioned in
Theorem~\ref{theorem-unions} are complete for $\Para[W_\forall]\Class{NL}$
and $\Para[W_\forall]\Class L$ \emph{except} for
$\PLang{union$_\forall$-tree}$, whose complexity remains open.

\paragraph{Associative Generability.}

The last union problem we study is based on the \textsc{generators}
problem, which contains tuples $(U,\circ,x,G)$ where $U$ is a set,
$\circ\colon U^2 \to U$ is (the table of) a binary operation, $x\in
U$, and $G \subseteq U$ is a set. The question is whether the
closure of $G$ under~$\circ$ (the smallest superset of $G$ closed
under~$\circ$) contains~$x$. A restriction of this problem is
$\Lang{associative-generator}$, where $\circ$ must be associative. By
two classical results, $\Lang{generators}$ is $\Class
P$-complete~\cite{JonesL1976} and $\Lang{associative-generator}$ is
$\Class{NL}$-complete~\cite{JonesLL1976}.

In order to apply the union operation to generator problems, we encode
$(U,\circ,x,G)$ as follows: $U$, $\circ$, and $x$ are encoded in some
sensible way using the alphabet~$\Sigma$. To encode $G$, we add a $1$
after the elements of $U$ that are in $G$ and we add a $0$ after
\emph{some} elements of $U$ that are not in~$G$. This means that in
the underlying templates we get the freedom to specify that only some
elements of~$U$ may be chosen for~$G$. Now,
$\PLang{weighted-union-generators}$ equals the problem known as
$\PLang{generators}$ in the literature: Given $\circ$, a
subset~$C\subseteq U$ of generator candidates, a  
parameter $k$, and a target element~$x$, the question is whether there
exists a set $G \subseteq C$ of size $|G| = k$ such that the closure
of~$G$ under $\circ$ contains~$x$. Flum and Grohe~\cite{FlumG2006}
have shown that $\PLang{generators}$ is complete for $\Class W[\Class
P] = \Para[W] \Class P$ (using a slightly different problem definition
that has the same complexity, however). Similarly, 
\PLangText{weighted-union-associative-generator} is also known as
$\PLang{agen}$ and we show: 

\begin{theorem}\label{theorem:agen}
  $\PLang{agen}$ is complete for $\Para[W]\Class{NL}$.
\end{theorem}

\begin{proof}
  Clearly, $\PLang{agen} \in \Para[WN]\Class L$ since the
  nondeterministic bits provided by the W-operator suffice to describe
  the generator set and since testing whether a set is, indeed, a
  generator is well-known to lie in $\Class{NL}$.

  For hardness note that $\Lang{agen}$ is complete for $\Class{NL}$
  under compatible logspace projections, see~\cite{JonesLL1976}. By
  Theorem~\ref{theorem-comp} we then have that
  \PLangText{family-union-agen} is complete for
  $\Para[W]\Class{NL}$. We now show that this problem reduces to
  \PLangText{subset-union-agen}, which in turn reduces to
  \PLangText{weighted-union-agen}, i.\,e., $\PLang{agen}$.

  \def\error{\hbox{\tt error}}%

  \medskip\emph{Hardness of \PLangText{subset-union-agen}.} For the first
  reduction let the compatible sets $S_1,\dots,S_k$ be given as input.
  The template encodes a universe~$U$, a set of generator
  candidates~$C\subseteq U$, a target element~$x\in U$, and an
  operation~$\circ\colon U^2\to U$. The instantiations encode subsets
  of the generator candidates. Our aim is to 
  construct a new instance of \PLangText{subset-union-agen}, i.\,e., a single
  set~$S'$ of compatible strings encoding a universe~$U'$, a set of
  generator candidates~$C'\subseteq U'$, a target element~$x'$, and an
  operation~$\circ'$ such that there are $k$~elements of~$S'$ that
  induce a generating set for~$x'$ if, and only if, there are $k$
  elements~$s_i$, one from every~$S_i$, such that they induce a
  generating set of~$x$. To achieve this we first set
  $U'=U\cup\{e_1,\dots,e_k\}$ for new elements~$e_i$, one for each
  $S_i$, and also add them to the new set of generator candidates
  $C'=C\cup\{e_1,\dots,e_k\}$. We then augment the operation~$\circ$
  to $\circ'$ with respect to the new elements requiring that no $e_i$
  can be generated by any combination of two other elements from the
  universe and that no $e_i$ can be used to generate
  elements from the universe other than itself (we achieve this by
  actually using whole string as elements of our universe, as will be
  discussed later in this proof). Furthermore, we insert a new target
  element~$x'$ into the  
  universe. Our aim is to enforce that~$x'$ can only be generated via
  the expression~$x\circ' e_1\circ' e_2\circ'\dots\circ' e_k$.
  Finally, we add an \emph{error element}~$\error$ to the universe
  that we will use to create dead ends in the evaluation of
  expressions: Any expression that does not make sense or contains the
  error element is evaluated to $\error$.
  
  The set $S'$ then contains a string~$s_{ij}'$ for every $s_{ij}\in
  S_i$ that is essentially $s_{ij}$ adjusted to $U'$, $C'$, $x'$, and
  $\circ'$, where we require that the binary string that selects a set
  of generators from $C'$ also selects $e_i$ and no other of the
  introduced elements $e_j$. From this we have that there is a
  selection of $k$ elements of $S'$ that induces a set of generators
  whose closure contains $x,e_1,\dots,e_k$ and therefore also $x'$ if,
  and only if, there is a set of $k$~strings $s_i\in S_i$ describing a
  set of generators whose closure contains~$x$.

  Unfortunately, our operator~$\circ'$ is binary and, therefore, we
  cannot evaluate expressions like $x\circ' e_1\circ'
  e_2\circ'\dots\circ' e_k$ in a single step. Moreover, because of the
  required associativity of~$\circ'$, it has to be possible to
  completely evaluate any subexpression of a larger expression. To
  achieve this, we actually use strings as elements of our
  universe, instead of single symbols, that are evaluated ``as far as
  possible.'' For instance, the expression $a\circ' b\circ' c\circ'
  e_1\circ' e_2$ evaluates to $d\circ' e_1\circ' e_2$ if the
  expression $a\circ' b\circ' c$ evaluates to~$d$. Since $d\circ'
  e_1\circ' e_2$ cannot be evaluated further, we want the string $d
  e_1 e_2$ to be part of our universe. 
  
  To formalize the idea of ``strings evaluated as far as possible,'' we need some
  definitions. Given an alphabet $\Gamma$, let us call a set $R$ of rules
  of the form $w \to w'$ with $w,w' \in \Gamma^*$ a
  \emph{replacement system}. An \emph{application} of a rule $w \to
  w'$ takes a word $u w v$ and yields the word $u w' v$; we write $u w v
  \Rightarrow_R u w' v$ in this case. A word is \emph{irreducible} if no
  rule can be applied to it. Let $\equiv_R$ be the reflexive,
  symmetric, transitive closure of $\Rightarrow_R$. Given a word $u$,
  let $[u]_R = \{v \mid u \equiv_R v\}$ be the equivalence class
  of~$u$. We use $\Gamma^* /_{\equiv_R} = \{ [v]_R \mid v \in
  \Gamma^*\}$ to denote the set of all equivalence classes of
  $\Gamma^*$. Observe that we can define a natural concatenation
  operation $\circ_R$ on the elements of $\Gamma^* /_{\equiv_R}$: Let $[u]_R \circ_R
  [v]_R = [u\circ v]_R$. Clearly, this operation is well-defined and
  associative. An \emph{irreducible representative system} of $R$ is a
  set of irreducible words that contains exactly one word from each
  equivalence class in $(U')^* /_{\equiv_R}$.
  
  In the context of our reduction, $\Gamma$ will be $U'$ and $R$ contains
  the following rules: First, for elements $a,b,c \in U$ of the original
  universe $U$ with $a \circ b = c$, we have the rule $ab \to
  c$. Second, we have the rule $x e_1 \dots e_k \to x'$. Third, we
  have the rules $\error u \to \error$ and $u \error \to \error$ for
  all $u \in U'$. Fourth, we have the rules $e_i u \to \error$ for all
  $u \in U' \setminus \{e_{i+1}\}$ and $x' u \to \error$ for all $u
  \in U'$.

  We can now, finally, describe the sets to which the reduction
  actually maps an input $(U,\circ,x,C)$: The universe $U''$ is an
  irreducible representative system of~$R$, the operation $\circ''$
  maps a pair $(u,v)$ to the representative of $[u\circ v]_R$, let the
  target element $x''$ be the representative of $[x']_R$, and let
  $C''$ contain all representatives of  $[c]_R$ for $c \in C'\}$. 
  
  Our first observation is that $(U')^* /_{\equiv_R}$ (and hence also
  $U''$) has polynomial size: Consider any 
  $[w]_R$ and let $w$ be irreducible. If $w$ does not happen to the error symbol
  itself, it cannot contain the error symbol (by the third
  rule). Furthermore, in~$w$ there cannot be any element from~$U$ to
  the right of any~$e_i$ or of~$x'$ (by the fourth rule). Thus, it
  must be of the form   $w_1 w_2$ with $w_1 \in U^*$ and $w_2 \in
  \{e_1,\dots,e_k,x'\}^*$. Then $w_1$ must actually be a single letter
  (by the first rule) and $w_2$ must by $x'$ or a sequence $e_i e_{i+1}
  \dots e_j$ for some $i \le j$ (by the fourth rule). This shows that
  the total number of different equivalence classes is at most $1+
  |U'| (k^2+1)$.

  The second observation concerns the equivalence class of $[x']_R$,
  which contains the string $x e_1\dots e_k$. We can only generate this
  class from elements $[c]_R$ with $c \in C'$ if these elements
  include all $[e_i]_R$ and also the equivalence classes $[c]_R$ of
  elements $c \in C$ that suffice to generate~$x$. This shows the
  correctness of the reduction.

  \def\g{%
    \hbox to 2.5pt{\hfil\vrule width.6pt height.125ex depth.125ex\hfil}%
  }%
  \def\error{\hbox{\tt error}}%
  \def\endsymbol{\hbox{$\triangleleft$}}%

  \medskip\emph{Hardness of \PLangText{weighted-union-agen}.} Given a
  compatible set $S=\{s_1,\dots,\penalty0 s_k\}$ whose strings encode
  a universe~$U$, a set of generator candidates~$C\subseteq U$, an
  associative operation~$\circ\colon U^2\to U$, and a target
  element~$x\in U$, together with a selection of generator candidates,
  we have to construct an instance~$S'$ such that every string only
  selects a single generator candidate. To achieve this we construct a
  new universe $U'$ that contains the elements of the old universe~$U$
  with new elements described below. As in the previous reduction, we
  define reduction rules alongside these new elements and then use an
  irreducible representative system of the rules as our universe.

  \begin{enumerate}
    
  \item We have an error element $\error$ with similar rules as above.

  \item We have an \emph{end element} $\endsymbol$. No rule has
    $\endsymbol$ on its right-hand side. Therefore, $\endsymbol$ has
    to be an  element of any generating set~$G$. We require this
    element for technical reasons that we will discuss later. There
    are rules $\endsymbol u \to \error$ for all $u \in U'$.

  \item We have a \emph{counter element} $\g$. Like the
    end symbol, this symbol cannot be generated by any expression and
    has to be an element of any generating set.

  \item We have elements $\sigma_i$ for each $s_i\in S$, which we call
    \emph{selector elements}. The idea behind these elements is that
    we will use them together with the counter element to enumerate
    all the elements $u_1,\dots,u_l$ of the generator candidates
    selected by a string $s_i\in S$. The objective is that strings like
    $\sigma_i\g\g\g\g$ can be replaced by $u_4$. We will give rules
    for this in a moment.

  \end{enumerate}
  
  In our new template, the candidates are (the representatives of the
  equivalence classes of) the $\sigma_i$ as well as $\endsymbol$,
  $\g$, and $\error$. Now, there is a selection of
  $k+3$ elements of $S'$ that forms a generating set for the target
  element if, and only if, there is a selection of $k$ elements of~$S$
  that forms a generating set. 

  It remains to explain how rules can be set up such that
  $\sigma_i\g\g\g\g$ gets replaced by $u_4$. Consider the expression
  $\sigma_1\g\g \sigma_1\g \sigma_3\g\g\g  \sigma_2\g\g$. Here,
  $\sigma_1\g\g$ can be replaced by some $u$ and $\sigma_3\g\g$ by
  some $u'$, but $ \sigma_2\g\g$ cannot yet be replaced since it is
  not clear what element will be appended to the expression (if there
  is another element). To fix this, we use the end symbol~$\endsymbol$
  that has to be appended to every expression. It marks the right end
  of the expression and enforces the unambiguous evaluation of the
  very last subexpression and, therefore, the whole expression.
  Translated into rules, this means that if, for instance,
  $\sigma_i\g\g\g\g$ should select $u_4$, then we have rules like
  $\sigma_i\g\g\g\g u \to u_4u$ if $u \neq \g$, but do not have the
  rule $\sigma_i\g\g\g\g \to u_4$.
  
  The target is the irreducible string $x\endsymbol$. 
  
  Again, the number of equivalence classes is polynomially in the
  size of the universe. Therefore, the reduction can be computed in
  $\Para\Class L$ space. 
\end{proof}

With the machinery introduced in this section, this result may not
seem surprising: \textsc{as\-so\-ci\-a\-tive-generators} is known to be
complete for $\Class{NL}$ via compatible logspace reductions and,
thus, by Theorem~\ref{theorem-comp},
\PLangText{family-union-as\-so\-ci\-a\-tive-gen\-er\-a\-tors} is complete for 
$\Para[W]\Class{NL}$. To prove Theorem~\ref{theorem:agen} we
``just'' need to further reduce to the weighted union
version. However, unlike for satisfiability and graph problems, 
this reduction turns out to be technically difficult.

\section{Problems Complete for Time--Space Classes}
\label{section-space-time}

The classes $\Para\Class P = \Class{FPT}$ and $\Class{XL}$ appear to be
incomparable: Machines for the first class may use $f_x n^{O(1)}$ time
and as much space as they want (which will be at most $f_x n^{O(1)}$),
while machines for the second class may use $f_x \log n$ space and as
much time as they want (which will be at most $n^{f_x}$). A natural
question is which problems are in the intersection
$\Para\Class P \cap \Class{XL}$ or -- even better -- in the class
$\Class D[f \operatorname{poly}, f \log]$, which means that there is a
\emph{single} machine using only fixed-parameter time and slice-wise 
logarithmic space simultaneously.

It is not particularly hard to find artificial problems that are
complete for the different time--space classes introduced in
Section~\ref{section-intro-ftpxp}; we present such problems at the
beginning of this section. We then move on to automata problems, but still
some ad hoc restrictions are needed to make the problems complete for
time--space classes. The real challenge lies in finding 
problems together with \emph{natural} parameterization that are
complete. We present one such problem: the longest common subsequence
problem parameterized by the number of strings. 

\paragraph{Resource-Bounded Machine Acceptance.}
A good starting point for finding complete problems for new classes is
typically some variant of Turing machine acceptance (or
halting). Since we study machines with simultaneous time--space 
limitations, it makes sense to start with the following ``time and
space bounded computation'' problems: For $\Lang{dtsc}$ the input is a
single-tape \textsc{dtm} $M$ together with two numbers $s$ and $t$
given in unary.  The question is whether $M$ accepts the empty string
making at most $t$ steps and using at most $s$ tape cells. The problem
$\Lang{ntsc}$ is the nondeterministic variant. As observed by Cai et
al.~\cite{CaiCDF1997}, the fpt-reduction closure of $\PLang[\mathit t]{ntsc}$
(that is, the problem parameterized by $t$) is exactly $\Class
W[1]$. In analogy, Guillemot~\cite{Guillemot2011} proposed the name
``$\Class{WNL}$'' for the fpt-reduction closure of $\PLang[\mathit
s]{ntsc}$ (now parameterized by~$s$ rather than~$t$). As pointed out in
Section~\ref{section-justification-w}, we 
believe that this name should be reserved for the class resulting from
applying the operator $\exists^{\leftrightarrow}_{f\log}$ to the class
$\Class{NL}$. Furthermore, the following theorem shows that
$\PLang[\mathit s]{ntsc}$ is better understood in terms of time--space
classes:

\begin{theorem}\label{theorem-dtsc}
  The problems $\PLang[\mathit s]{dtsc}$ and $\PLang[\mathit s]{ntsc}$ are
  complete for the classes $\Class D[f\operatorname{poly},f\log]$
  and $\Class N[f\operatorname{poly},f\log]$, respectively. 
\end{theorem}

\begin{proof}
  We only prove the claim for the deterministic case, the
  nondeterministic case works exactly the same way. For containment,
  on input of a machine $M$, a time bound $t$ in unary, and a space
  bound $s$ in unary, a  $\Class
  D[f\operatorname{poly},f\log]$-\penalty0machine can simply simulate
  $M$ for $t$ steps, making sure that no more than $s$ tape cells are
  used. Clearly, the time needed for this simulation is a fixed
  polynomial in $t$ and, hence, in the input length. The space needed
  to store the $s$ tape cells is clearly $O(s \log n)$ since $O(\log
  n)$ bits suffice to store the contents of a tape cell (the amount
  needed is not $O(1)$ since the tape alphabet is part of the input).

  For hardness, consider any problem $(Q,\kappa) \in \Class
  D[f\operatorname{poly},f\log]$ via some machine~$M$. Let $t_M(k,n)$
  and $s_M(k,n)$ be the time and space bounds of $M$, respectively.
  The reduction must now map inputs $x$ to triples $(M',1^t,1^s)$. The
  reduction faces two problems: First, while $M$ has an input tape and
  a work tape, $M'$ has no input tape and starts with the empty
  string. Second, while $t$ can simply be set to $t_M(\kappa(x),x)$,
  $s$ cannot be set to $s_M(\kappa(x),x)$ since this only lies in
  $O\bigl(f(\kappa(x)) \log 
  |x|\bigr)$ for some function~$f$ -- while in a parameterized reduction the
  new parameter may only depend on the old one ($\kappa(x)$) and not
  on the input length. The first problem can be overcome
  using a standard trick: $M'$ simulates $M$ and uses its tape to
  store the contents of the work tape of~$M$. Concerning the input
  tape (which $M'$ does not have), when $M$ accesses an
  input symbol, $M'$ has this symbol ``stored in its state,'' which
  means that there are $|x|$ many copies of $M$'s state set inside
  $M'$, one for each possible position of the head on the input
  tape. A movement of the head corresponds to switching between these
  copies. In each copy, the behaviour of the machine $M$ for the
  specific input symbol represented by this copy is hard-wired.

  The second problem is a bit harder to tackle, but one can also apply
  standard tricks. Instead of mapping to~$M'$, we actually map to a new
  machine $M''$ that performs the following space compression trick: For each
  $\log_2 |x|$ many tape cells of $M'$, the machine $M''$ uses only
  one tape cell. This can be achieved by enlarging the tape alphabet
  of $M'$: If the old alphabet was $\Gamma$, we now use
  $\Gamma^{\log_2 |x|}$, which is still polynomial in
  $|x|$. Naturally, we now have to adjust the transitions and states
  of $M'$ so that a step of $M'$ for its old tape is now appropriately
  simulated by one step of $M''$ for its compressed tape.

  Taking it all together, we map $x$ to $(M'', t,s)$ where $t$ is as
  indicated above and $s = s_M(\kappa(x),x) / \log_2 |x|$, which is
  bounded by a function depending only on $\kappa(x)$. Clearly, the
  reduction is correct.
\end{proof}

\paragraph{Automata.} A classical result of
Hartmanis~\cite{Hartmanis1972} states that $\Class L$ contains exactly 
the languages accepted by finite multi-head automata. In
\cite{ElberfeldST2012}, Elberfeld et al.\ used this to show that
$\PLang[heads]{mdfa}$ (the multi-head automata acceptance 
problem parameterized by the number of heads) is complete for
$\Class{XL}$. It turns out that multi-head automata can also be used
to define a (fairly) natural complete problem for $\Class D[f
\operatorname{poly}, f\log]$: A \emph{\textsc{dag}-automaton} is an
automaton whose transition graph is a
topologically sorted \textsc{dag} (formally, the states must form the
set $\{1,\dots, |Q|\}$ and the transition function must map each
state to a strictly greater state). Clearly,
a \textsc{dag}-automata will never need more than $|Q|$ steps to
accept a word, which allows us to prove the following theorem: 

\begin{theorem}\label{theorem-multi}
  The problems $\PLang[heads]{dag-mdfa}$ and $\PLang[heads]{dag-mnfa}$
  are complete for the classes $\Class D[f
  \operatorname{poly}, f\log]$ and $\Class N[f \operatorname{poly},
  f\log]$, respectively.
\end{theorem}

\begin{proof}
  We only prove the claim for deterministic automata, the argument
  works exactly the same way for nondeterministic ones. Let
  $(Q,\kappa) \in \Class D[f \operatorname{poly},f \log]$ via
  some~$M$. Then $(Q,\kappa) \in \Class{XL}$ will hold via the same 
  machine~$M$. In \cite{ElberfeldST2012} it is shown that every problem in
  $\Class{XL}$ can be reduced to $\PLang[heads]{mdfa}$ via a simulation
  dating back to the work of Hartmanis~\cite{Hartmanis1972}: Each step
  of~$M$ is simulated by a number of movement of the heads of an
  automaton~$A$. The positions of a fixed number of heads of~$A$ 
  store the contents of one work tape symbol and for each step of~$M$
  the heads of~$A$ perform a complicated ballet to determine the current
  contents of certain tape cells and to adjust the heads accordingly. 

  For our purposes, it is only important that each step of $M$ gives
  rise to a polynomial number (in the input length) of steps of $A$
  for some fixed polynomial independent of the number of heads. In
  particular, to simulate $f(\kappa(x)) n^c$ steps of~$M$, the
  automaton~$A$ needs to perform $f'(\kappa(x)) n^{O(c)}$ 
  steps. Thus, to reduce $(Q,\kappa)$ to $\PLang[heads]{dag-mdfa}$, we
  first compute $A$ as in the reduction to $\PLang[heads]{mdfa}$, but
  make $f'(\kappa(x)) n^{O(c)}$ copies of~$A$. We then modify the
  transitions such that when a transition in the $i$th copy $A$ maps a
  state $q$ to a state $q'$, the transition then instead maps this $q$ to
  $q'$ from the $(i+1)st$ copy, instead. Clearly, the resulting
  automaton is a \textsc{dag}-automaton and it accepts an input word
  if, and only if, $M$ accepts it in time $f(\kappa(x)) n^c$ and using
  only $f(\kappa(x)) \log n$ space.
\end{proof}

Instead of \textsc{dag}-automata, we can also consider a ``bounded
time version'' of $\Lang{mdfa}$ and $\Lang{mnfa}$, where we ask whether
the automaton accepts within $s$ steps ($s$ being given in
unary). Both versions are clearly equivalent: 
The number of nodes in the \textsc{dag} bounds the number
of steps the automaton can make and cyclic transitions graphs can be
made acyclic by making $s$ layered copies.

Another, rather natural kind of automata are \emph{cellular automata,}
where there is one instance of the automaton (called a \emph{cell})
for each input symbol. The cells perform individual synchronous computations, but
``see'' the states of the two neighbouring cells (we only
consider one-dimensional automata, but the results hold for any fixed
number of dimensions). Formally, the transition function of such an
automaton is a function $\delta \colon Q^3 \to Q$ (for the cells at
the left and right end this has to be modified appropriately). The
``input'' is just a string $q_1\dots q_k \in Q^*$ of states and the
question is whether $k$ cells started in the states $q_1$ to $q_k$
will arrive at a situation where one of them is in an accepting state
(one can also require all to be in an accepting state, this makes no
difference).

Let $\Lang{dca}$ be the language $\bigl\{(C,q_1\dots q_k) \mid C$ is a 
deterministic cellular automaton that accepts $q_1\dots
q_k\smash{\bigr\}}$. Let $\Lang{nca}$ denote the nondeterministic version and let 
$\Lang{dag-dca}$ and $\Lang{dag-nca}$ be the versions where $C$ is
required to be a \textsc{dag}-auto\-ma\-ton (meaning that $\delta$ must
always output a number strictly larger than all its inputs). The
following theorem states the complexity of the resulting problems when
we parameterize by $k$ (number of cells): 

\begin{theorem}\label{theorem-cellular}
  The problems $\PLang[cells]{dca}$ and $\PLang[cells]{nca}$ are
  complete for $\Class{XL}$ and $\Class{XNL}$, respectively. The
  problems $\PLang[cells]{dag-dca}$ and $\PLang[cells]{dag-nca}$ are
  complete for $\Class D[f \mathrm{poly},f \log]$ and $\Class N[f
  \mathrm{poly},f \log]$, respectively.
\end{theorem}

\begin{proof}
  We start with containment and then prove hardness for all problems.
  
  \medskip\emph{Containment.} Clearly, $\PLang[cells]{dca}$ lies in 
  $\Class{XL}$ since we can keep track of the $k$ states of the $k$
  cells in space $O(k \log n)$. To see that
  $\PLang[cells]{dag-dca} \in \Class D[f \operatorname{poly},\penalty100 f\log]$,
  just observe that for \textsc{dag}-automata no computation can take
  more than a linear number of steps. The arguments for the
  nondeterministic versions are the same. 
  
  \medskip\emph{Hardness for Deterministic Cellular Automata.} To
  prove hardness of $\PLang[cells]{dca}$ for $\Class{XL}$, we reduce
  from a canonically complete problem for 
  $\Class{XL}$. Such a problem can easily be obtained from
  $\PLang[\mathit s]{dtsc}$ by 
  lifting the restriction on the time allowed to the machine, leading
  to the following problem:
  
  \begin{problem}[{$\PLang[\mathit s]{deterministic-space-bounded-computation}$
    ($\PLang[\mathit s]{dsc}$)}]
    \begin{parameterizedproblem}
      \instance (The code of) a single-tape machine $M$, a number $s$.
      \parameter $s$.
      \question Does $M$ accept on an initially empty tape using at
      most $s$ tape cells?
    \end{parameterizedproblem}
  \end{problem}

  Proving that this problem is complete for $\Class{XL}$ follows
  exactly the same argument as that used in
  Theorem~\ref{theorem-dtsc}.

  Let us now reduce $\PLang[\mathit s]{dsc}$ to
  $\PLang[cells]{dca}$. The input for the reduction is a pair  
  $(M,s)$. We must map this to some cellular automaton $C$ and an
  initial string of states. The obvious idea is to have one automaton
  for each tape cell that can be   used by~$M$. In detail, let $Q$ be
  the set of states of $M$ and let 
  $\Gamma$ be the tape alphabet of $M$. The state set of $C$ will be
  $R = (Q \cup \{\bot\}) \times \Gamma$, where $\bot$ is used to
  indicate that the head is elsewhere. Clearly, a state string from
  $R^s$ allows us to encode a configuration of~$M$. Furthermore, we
  can now set up the transition relation of~$C$ 
  in such a way that one parallel step of the $s$ automata corresponds
  exactly to one computational step of~$M$: as an example, suppose in
  state $q$ for the symbol $a$ the machine $M$ will not move its head,
  write~$b$, and switch to state $q'$. Then in $C$ for every $x,y \in
  \{\bot\} \times \Gamma$ there would be a transition mapping $(x,
  (q,a), y)$ to $(q',b)$ and also transitions mapping $((q,a),x,y)$ 
  to~$x$ and $(x,y,(q,a))$ to~$y$. For triples corresponding to
  situations that ``cannot arise'' like the head being in two places
  at the same time, the transition function can be setup arbitrarily. 
  The initial string of states for the cellular automaton is of course
  $(q_0,\Box)(\bot,\Box)\dots(\bot,\Box)$, where $\Box$ is the
  blank symbol and $q_0$ is the initial state of~$M$.

  With this setup, the strings of states of
  the cellular automaton are in one-to-one correspondence with the
  configurations of~$M$. In particular, we will reach a state string
  containing an accepting state if, and only if, $M$ accepts when
  started with an empty tape. Clearly, the reduction is a para-$\Class
  L$-reduction.

  \medskip\emph{Hardness for Nondeterministic Cellular Automata.} One
  might expect that one can use the exact same argument for
  nondeterministic automata and simply use the same reduction, but
  starting from $\PLang[\mathit s]{nsc}$. However, there is a
  complication: The cells work independently of one another. In
  particular, there is no guarantee that a nondeterministic decision
  taken by one cell is also 
  take by a neighboring cell. To illustrate this point, consider the
  situation where the machine $M$ can
  nondeterministically step ``left or right'' in some state~$q$. Now
  assume that some cell $c$ is in state $(q,x)$ and consider the cells
  $c-1$ and $c+1$. For both of them, there would now be a transition
  allowing them to ``take over the head'' and both could
  nondeterministically decide to do so -- which is wrong, of course;
  only one of them may receive the head.

  To solve this problem, we must ensure that a nondeterministic
  decision is taken ``by only one cell.'' Towards this aim, we
  first modify~$M$, if necessary, so that every nondeterministic
  decision is a binary decision. Next, we change the state set of~$C$: 
  Instead of $(Q \cup \{\bot\}) \times \Gamma$ we use $(Q \times
  \{0,1\} \cup Q \cup \{\bot\}) \times
  \Gamma$. In other words, in addition to the normal states from $Q$
  we add two copies of the state set, one tagged with~$0$ and one
  tagged with~$1$. The idea is that when a cell is in state $(q,x) \in
  Q \times \Gamma$, it can nondeterministically reach $((q,0),x)$ or
  $((q,1),x)$. However, from those states, we can
  \emph{deterministically} make the next step: if the state is tagged
  by $0$, both the cell and the neighboring cells continue according
  to what happens for the first of the two possible nondeterministic
  choices, if the state is tagged by $1$, the other choice is
  used. Note that as long as the state is not yet tagged, the
  neighboring cells do not change their state.

  With these modifications, we arrive at a new cellular automaton with
  the property that after every two computational steps of the
  automaton its string of states encodes one of the two possible next
  computational steps of the machine~$M$. This shows that the reduction
  is correct.

  \medskip\emph{Hardness for Cellular \textsc{dag}-Automata.} To prove
  hardness of $\PLang[cells]{dag-dca}$ for the class $\Class D[f
  \operatorname{poly}, f\log]$, we reduce form $\PLang[\mathit
  s]{dtsc}$, which is complete for the class by
  Theorem~\ref{theorem-dtsc}. On input $(M,1^t,1^s)$, the reduction is
  initially exactly the same as for $\PLang[cells]{dca}$ and we just
  ignore the time bound~$t$. Once an automaton $C$ has been computed,
  we can turn it into a \textsc{dag}-automaton and incorporate the
  time bound as follows: We create $t$ many copies of $C$ and
  transitions that used to be inside one copy of~$C$ now lead to the
  next copy (this is same idea as in the proof of
  Theorem~\ref{theorem-multi}). This construction ensures that the
  automaton will accept the initial sequence if, and only if, $M$
  accepts on an empty input tape in time~$t$ using space~$s$.

  For the nondeterministic case, we combine the constructions we
  employed for $\PLang[cells]{dag-dca}$ and for $\PLang[cells]{nca}$.
\end{proof}

We remark that, for once, the nondeterministic cases need special arguments.

\paragraph{Pebble Games.}

Pebble games are played on graphs on whose vertices we place pebbles 
(a \emph{pebbling} is thus a subset of the set of 
vertices) and, over time, we (re)move and add pebbles according to 
different rules. Depending on the rules and the kind of graphs, the
resulting computational problems are 
complete for different complexity classes, which is why pebble games
have received a lot of attention in the literature. We 
introduce a simple pebble game played by a single player, different
versions of which turn out to be complete for different parameterized
space complexity classes: A \emph{threshold pebble game
  (\textsc{tpg})} consists of a directed  graph $G= (V,E)$ together
with a threshold function $t \colon V \to \mathbb N$. Given a pebbling
$X \subseteq V$, a vertex~$v$ \emph{can be pebbled after~$X$} if the
number of $v$'s pebbled predecessors is at least $v$'s threshold, that
is, $\bigl| \{\, p \mid (p,v) \in E\} \cap X\bigr| \ge 
t(v)$. Given a 
pebbling~$X$, \emph{a next pebbling} is any set $Y$ of vertices that can
be pebbled after~$X$. The \emph{maximum next pebbling} is the maximum
of such~$Y$. 

The language $\Lang{tpg}$ contains all threshold pebble games together
with two pebblings $S$ and $T$ such that we can reach $T$ when we
start with $S$ and apply the next pebbling operation repeatedly
(always replacing the current pebbling $X$ completely by~$Y$). For the
$\Lang{tpg-max}$ problem, $Y$ is always chosen as the maximum next
pebbling (which makes the game deterministic). For $\Lang{dag-tpg}$
and $\Lang{dag-tpg-max}$, the graph is restricted to be a
\textsc{dag}. In the following theorem, we parameterize by  the
maximum number of pebbles that may be present in any step. 

\begin{theorem}\label{theorem-tpg1}
  The problems $\PLang[pebbles]{tpg-max}$  and $\PLang[pebbles]{tpg}$
  are complete for $\Class{XL}$ and $\Class{XNL}$, respectively.  
  The problems $\PLang[pebbles]{dag-tpg-max}$  and $\PLang[pebbles]{dag-tpg}$
  are complete for $\Class D[f \mathrm{poly},f \log]$
  and $\Class N[f \mathrm{poly},f \log]$,  respectively. 
\end{theorem}

\begin{proof}
  We first prove containment for all problems. Then we prove
  completeness first for the \textsc{dag} versions and then for the
  general version.
  
  \medskip\emph{Containment.} To see that
  $\PLang[pebbles]{tpg-max}\in\Class{XL}$ holds, observe that a
  deterministic Turing machine can store a pebbling in space $\kappa(x)
  O(\log n)$. Starting with $S$, the machine can compute the
  successive next steps and accepts when $T$ is reached. Similarly,
  $\PLang[pebbles]{tpg} \in \Class{XNL}$ since we can nondeterministically guess
  the correct subset of the next pebbling that has to be chosen. For
  the problems for \textsc{dag}s, observe that by the acyclicity the
  maximal distance of a pebble to any sink in the graph gets reduced
  by~$1$ in every step. Since the maximal distance at the start of the
  simulation is bounded by~$|V|$, we must reach $T$ after at most
  $|V|$ steps and, thus, the simulation can be done in both time
  $k\cdot |V|^{O(1)}$ and space $k\cdot O(\log|V|)$. 

  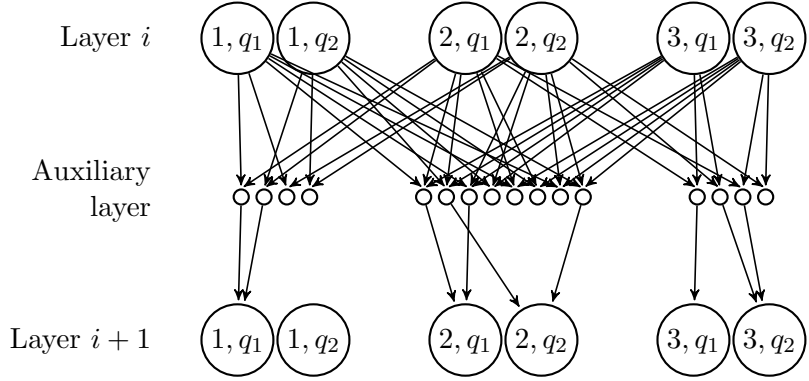
\begin{figure}[htpb]
    \begin{tikzpicture}[anchor=mid]
      \node [anchor=mid east] at (-1,4) {Layer $i$};
      \node [anchor=mid east] at (-1,0) {Layer $i+1$};
      \node [align=right,left] at (-1,2) {Auxiliary\\ layer};
      \graph [graph, no placement, nodes=node] {
        { [y=4]
          s11 / ${1,q_1}$ [x=0],
          s12 / ${1,q_2}$ [x=1],
          s21 / ${2,q_1}$ [x=3],
          s22 / ${2,q_2}$ [x=4],
          s31 / ${3,q_1}$ [x=6],
          s32 / ${3,q_2}$ [x=7]
        },
        { [y=0]
          t11 / ${1,q_1}$ [x=0],
          t12 / ${1,q_2}$ [x=1],
          t21 / ${2,q_1}$ [x=3],
          t22 / ${2,q_2}$ [x=4],
          t31 / ${3,q_1}$ [x=6],
          t32 / ${3,q_2}$ [x=7]
        },
        { [nodes={as=,minimum size=3pt,text height=,text depth=}, y=2]
          a1 [x=0.05], a2[x=0.35], a3[x=0.65], a4[x=0.95]
        },
        { [nodes={as=,minimum size=3pt,text height=,text depth=}, y=2]
          b1 [x=2.45], b2[x=2.75], b3[x=3.05], b4[x=3.35],
          b5 [x=3.65], b6[x=3.95], b7[x=4.25], b8[x=4.55],
        },
        { [nodes={as=,minimum size=3pt,text height=,text depth=}, y=2]
          c1 [x=6.05], c2[x=6.35], c3[x=6.65], c4[x=6.95]
        },
        { [right anchor=north]
          {s11,s21} -> a1,
          {s12,s21} -> a2,
          {s11,s22} -> a3,
          {s12,s22} -> a4,
        },
        { [right anchor=north]
          {s11,s21,s31} -> b1,
          {s12,s21,s31} -> b2,
          {s11,s22,s31} -> b3,
          {s12,s22,s31} -> b4,
          {s11,s21,s32} -> b5,
          {s12,s21,s32} -> b6,
          {s11,s22,s32} -> b7,
          {s12,s22,s32} -> b8,
        },
        { [right anchor=north]
          {s21,s31} -> c1,
          {s22,s31} -> c2,
          {s21,s32} -> c3,
          {s22,s32} -> c4,
        },
        { 
          {a1,a2} -> t11,
          {b1,b3} -> t21,
          {b2,b8} -> t22,
          {c1}    -> t31,
          {c2,c3} -> t32
        }        
      };
    \end{tikzpicture}
    \caption{Example of two layers of a threshold game constructed for
      $Q =  \{q_1,q_2\}$ and $s = 3$. The threshold of the large
      vertices is $1$, the threshold of the small vertices is~$2$ at
      the borders and~$3$ in the middle. Each pair of large vertices
      represents the two states a cell can be in. Note how each
      auxiliary vertex can be pebbled if, and only if, a certain pair
      or triple of states is reached by the cells in layer~$i$. The
      two arrows leading into the vertex $(1,q_1)$ at the bottom mean
      that the first cell can switch to state $q_1$ if it was in any
      state before, but its neighboring cell was in state~$q_1$. The
      fact that there is no edge to $(1,q_2)$ means that the first cell will
      not switch to state~$q_2$, regardless of its earlier state and
      that of its neighbor. 
    }
    \label{fig-tpg}
  \end{figure}
  
  \medskip\emph{Completeness of the maximum \textsc{dag}-version.} We
  reduce $\PLang[cells]{dag-dca}$ to the problem $\PLang[pebbles]{dag-tpg-max}$,
  which proves the claim by Theorem~\ref{theorem-cellular}. Let
  $(C,s)$ be given as input for the reduction and let $Q$ be
  $C$'s state set. The reduction outputs a pebble graph that encodes
  the computation of $s$ many copies of $C$ in the following way: It
  consists of $t = |Q|$ layers, each encoding a configuration during
  the computation. Each layer consists of $s$ blocks of $|Q|$ vertices
  and placing one pebble in each block clearly encodes exactly one
  state string. 

  To connect two layers $L$ and $L'$, we first insert an
  auxiliary layer between them with $s \cdot |Q|^3$ vertices, namely
  one vertex for each cell and each triple of a state of the preceding
  cell, in the own cell, and in the next cell.
  The threshold of the auxiliary vertex is $3$.  (Actually, for the
  first and last cell, only $|Q|^2$ auxiliary vertices are needed and
  the threshold is $2$. However, to simplify the presentation, we
  ignore these special cells in the following).  Note that in a game step,  
  an auxiliary vertex can be pebbled if, and only if, the previous,
  current, and next cell were in specific states. Also note that
  there will be exactly $s$ vertices on each auxiliary level that can
  be pebbled in such a step if the level before it corresponded to a
  configuration. 

  We now connect the vertices of the auxiliary layer to~$L'$. Let
  $(c,q)$ be a vertex of $L'$, corresponding to a cell position $c$
  and a state $q$. Its predecessors will be all auxiliary vertices
  $(c,q_1,q_2,q_3)$ of the preceding layer such that $C$'s transition
  function maps $(q_1,q_2,q_3)$ to $q$. The threshold of all layer
  vertices is~$1$. Figure~\ref{fig-tpg} depicts an
  example of this construction. Since the auxiliary vertices can be
  pebbled exactly 
  if the automaton reaches the three states $(q_1,q_2,q_3)$, exactly
  those layer vertices can be pebbled that correspond to the next
  configuration of the cellular automaton. 
  
  To conclude the description of the reduction, let $S$ be the
  pebbling placing vertices on the first layer corresponding to the
  initial state string (which is part of the input) and let $T$ be the
  pebbling placing vertices on the last layer corresponding to the
  (only) accepting configuration (this can be achieved by an
  appropriate modification of $C$, if necessary).

  By construction, the deterministic game played on the graph starting
  with $S$ will end in $T$ if, and only if, the $s$-cell version of~$C$
  accepts the state sequence after $t$ steps. This shows that the
  reduction is correct. Clearly, the reduction is a para-$\Class
  L$-reduction. 
  
  \medskip\emph{Completeness of the \textsc{dag}-version.}
  For the completeness  for $\Class N[f
  \mathrm{poly},f \log]$ of the problem
  $\PLang[pebbles]{dag-tpg}$ we use the same
  reduction as above, but start from
  $\PLang[cells]{dag-nca}$. It is not immediately obvious that this
  reduction is correct since in a non-maximal pebble game the
  nondeterminism of the game allows us to ``forget'' pebbles and,
  possibly, this could ``make room'' for ``illegal'' pebble to appear
  that could disrupt the whole simulation.

  To see that this does not happen, let us introduce the following
  notion: For a cell number~$c$ let us say that a vertex
  \emph{belongs} to $c$ if it is either a vertex on the main layers of
  the form $(c,q)$ or a vertex on an auxiliary layer of the form
  $(c,q_1,q_2,q_3)$. Clearly, every vertex belongs to exactly one cell
  and in the deterministic case in each step for each cell exactly one
  vertex belonging to this cell is pebbled.

  The crucial observation is that if in a non-maximal step we do not
  pebble any vertex of a cell~$c$, we cannot pebble any vertex of
  cell~$c$ in any later step. This is due to the way the edges are set
  up: In order to pebble a vertex belonging to a cell~$c$ it is always
  a prerequisite that at least one vertex belonging to $c$ was pebbled
  in the previous step.

  Now, in the target pebbling $T$ for each cell one vertex belonging
  to this cell is pebbled. By the above observation, in order to reach
  $T$ this must have been the case in all intermediate steps. Because
  of the upper bound of $s$ on the number of pebbles, we know that on
  each layer \emph{exactly} $s$ vertices are pebbled and, thus,
  exactly one vertex belonging to each cell is pebbled. This proves that,
  indeed, on each main layer the pebbled vertices correspond exactly
  to possible cell contents of a computation of the automaton.

  \medskip\emph{Completeness of the maximal general version.} 
  Let us now prove that the problem $\PLang[pebbles]{tpg-max}$ is
  complete for $\Class{XL}$ by reducing from $\PLang[cells]{dca}$. We
  basically proceed as for the \textsc{dag} version, but no longer need
  to produce an acyclic graph and no longer have a time bound in the
  input. The idea is to let the computation continue ``as long as
  necessary'': Instead of constructing a graph consisting of
  many identical  
  main layers that alternate with identical auxiliary layers, we put
  only \emph{one} main layer in the graph and only \emph{one}
  auxiliary layer. The predecessors of the auxiliary layer's vertices
  are, as before, the vertices of the main layer. However, the
  successors of the auxiliary layer's vertices are no longer the
  vertices on the next main layer, but the corresponding vertices on
  the first (and only) main layer.

  With this construction, the pebbling will alternate between the main
  layer and the auxiliary layer and, after every two game steps, the
  main layer will encode the next configuration of the input
  machine~$M$. By having $S$ encode the start  
  configuration and $T$ encode the only accepting configuration (on
  the main layer), we get a reduction.

  \medskip\emph{Completeness of the general version.}
  For the general version, as for the \textsc{dag} version, one can
  argue that the same reduction as for the maximal case works for the
  general case.
\end{proof}

Another natural parametrization is by the number of steps needed
rather than by the number of pebbles. It is easily 
seen that $\PLang[steps]{tpg} \in \Para\Class{NP}$ and
$\PLang[steps]{tpg-max} \in \Para\Class P$. Furthermore,
$\PLang[steps]{tpg-max} \in \Class{XL}$ holds also since we can
compute one next step in the game in space $O(\log n)$ and by the standard
trick of chaining together two logspace computations, we can compute
$k$ steps in space $O(k \log n)$. Interestingly, the
argument can neither be used to show that 
$\PLang[steps]{tpg}$ lies in $\Class{XNL}$ nor that
$\PLang[steps]{tpg-max}$ lies in $\Class
D[f\operatorname{poly},f\log]$. We were not able to 
prove completeness of either problem for a parameterized class.

One can also consider a ``power parameterization'' similar to that of
$\PLang{power-ntsc}$: in $\PLang[step\ power]{tpg}$ we are 
given a parameter $k$ along with a threshold pebble game and ask
whether $T$ can be reached from $S$ in $n^k$ steps, where $n$ is the
order of the graph. As for the generic Turing machine problem, the
power parametrization results in problems that are complete for
$\Class N[n^f,f \mathrm{poly}]$ (for $\PLang[step\ power]{tpg}$) and
for $\Class D[n^f,f \mathrm{poly}]$ (for $\PLang[step\
power]{tpg-max}$). The proof is essentially the same as in the above
theorem, only the machine can now use much more space ($f_x\cdot n^c$
for some constant~$c$ instead of $f_x\cdot O(\log n)$), but we can
also use more pebbles (up to $n$ many instead of just $k$).

\paragraph{Longest Common Subsequence.}

\label{section-lcs}

The input for the longest common subsequence problem $\Lang{lcs}$ is a
set $S$ of strings over some alphabet $\Sigma$ together with a
number~$l$. The question is whether there is a string $c \in \Sigma^l$
that is a \emph{subsequence} of all strings in $S$, meaning that for
all $s \in S$ just by removing symbols from $s$ we arrive at~$c$.

There are several natural parameterization of $\Lang{lcs}$: We can
parameterize by the number of strings in~$S$, by the size of
the alphabet, by the length~$l$, or any combination thereof. Guillemot 
has shown \cite{Guillemot2011} that $\PLang[strings,length]{lcs}$ is
fpt-complete for $\Class W[1]$, while $\PLang[strings]{lcs}$ is
fpt-equivalent to $\PLang[\mathit s]{ntsc}$. Hence, by
Theorem~\ref{theorem-dtsc}, both problems are complete  
under fpt-reductions for the fpt-reduction closure of $\Class N[f 
\operatorname{poly}, f\log]$. We tighten this in
Theorem~\ref{theorem-lcs} below (using a weaker reduction is more than
a technicality: $\Class N[f \operatorname{poly}, f\log]$ is presumably
not even closed under fpt-reduction, while it \emph{is} closed under
para-$\Class L$-reductions).

As a preparation for the proof of  Theorem~\ref{theorem-lcs}, we first
present a simpler-to-prove result: Let $\Lang{lcs-injective}$ denote the
restriction of $\Lang{lcs}$ where all input words must be
\emph{$p$-sequences} \cite{FellowsHS2003}, which are words containing
any symbol at most once (the function mapping word indices to word
symbols is injective).

\begin{theorem}\label{theorem-lcs-injective}
  $\Lang{lcs-injective}$ is $\Class{NL}$-complete and this holds
  already under the restriction $|S| \le 4$.
\end{theorem}

\begin{proof}
  The problem $\Lang{lcs-injective}$ lies in $\Class{NL}$ via the
  following algorithm: We guess the common subsequence $c$
  nondeterministically and use a counter to ensure that $c$ has length
  at least~$l$. The problem is that we cannot remember more than a
  fixed number of letters of~$c$ without running out of
  space. Fortunately, we do not need to: We always only keep track of
  the last two guessed symbols. For each such pair $(a,b)$, we check
  whether $a$ appears before $b$ in all strings in $S$. If so, we move
  on to the next pair, and so on. Clearly, this algorithm needs only
  logarithmic space and correctly decides $\Lang{lcs-injective}$.

  To prove hardness for $|S| = 4$, we reduce from the
  $\Class{NL}$-complete language $\Lang{layered-reach}$, where the
  input is a layered graph $G$ (each vertex is assigned a layer number
  and   all edges only go from one layer to the next), the source
  vertex $s$ is the (only) vertex on layer~$1$ and the target $t$ is
  the (only) vertex on the last layer~$m$. The question is whether
  there is a path from $s$ to~$t$.

  For the reduction to $\Lang{lcs-injective}$ we introduce a symbol
  for each edge of~$G$. The common subsequence will then be exactly
  the sequence of edges along a path from $s$ to~$t$.   
  We consider the layers $L_1$, $L_2$, \dots, $L_m$ in order and, for
  each of them, append edge symbols to the four strings as described
  in the following.

  Consider a layer $L_i$, containing vertices $\{v_1, \dots, v_n\}$.
  Assume $i$ is odd. We go over the vertices $v_1$ to $v_n$ in that
  order. For $v_1$, first consider all edges that end at
  $v_1$. They must come from layer $i-1$. We add these edges in some
  order to the first string (for instance, in the order of the index
  of the start vertex of these edges). Still considering $v_1$, we
  then consider all outgoing edges and append them in some fixed
  order. Then we move on to $v_2$ and add edge symbols in the same way
  for it, and so on. If $i$ is even rather than odd, we add the same
  edge symbols to the third rather than to the first string.

  For the second (or, for even $i$, the fourth string), we go over the
  vertices in decreasing order. We start with $v_n$. We consider the
  incoming edges for $v_n$ and add them to the second string, but in 
  reverse order compared to the order we used for the first
  string. Next, we append the outgoing edges, again in reverse
  order. Then we consider $v_{n-1}$ and proceed in the same way.

  As an example, consider the following layered graph:
  \medskip
  
  
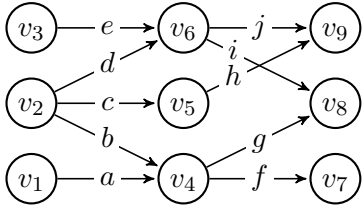
\begin{tikzpicture}
    \graph[graph, no placement, typeset=$v_\tikzgraphnodetext$, nodes=node] {
      { [x=0] 1[y=1], 2[y=2], 3[y=3]},
      { [x=2] 4[y=1], 5[y=2], 6[y=3]},
      { [x=4] 7[y=1], 8[y=2], 9[y=3]};

      1 ->["$a$"] 4 -> ["$f$"] 7;
      2 ->["$b$"] 4 -> ["$g$"] 8;
      2 ->["$c$"] 5 -> ["$h$" near start] 9;
      2 ->["$d$"] 6 -> ["$i$" near start] 8;
      3 ->["$e$"] 6 -> ["$j$"] 9;
    };
  \end{tikzpicture}
  \medskip
  
  This would result in the following strings, where spaces have been
  added for clarity and also the symbols $v_i$, which are not part of
  the strings (so the second string is actually $edcbajhigf$):
  \def\v#1{\textcolor{black!50}{v_#1}}
  \begin{align*}
    &\v1 a\quad \v2 bcd\quad \v3 e\qquad f\v7\quad gi\v8\quad hj\v9\\
    &\v3 e\quad \v2 dcb\quad \v1 a\qquad jh\v9\quad ig\v8\quad f\v7\\
    &ab \v4 fg\quad c \v5 h\quad de \v6 ij\\
    &ed \v6 ji\quad c \v5 h\quad ba \v4 gf  
  \end{align*}
  
  We make two crucial observations. First, if an edge is included in the
  common subsequence, no other edge starting at the same layer can be
  included also: The edge symbols of one layer come in one order in
  the first (or third) string and in the reverse order in the second
  (or fourth) string. Thus, there cannot be two of them in the common
  subsequence. For the same reason, there can only be one edge
  arriving at a layer in the common subsequence. The second crucial
  observation is that if the sequence contains an edge $e$ arriving at a
  vertex $v$, it can only contain edges leaving from vertex $v$, if it
  contains any edge leaving from $v$'s layer: Only the edge  $e'$ leaving
  $v$ will come after $e$ in both the first and second (or third and
  fourth) string.

  Putting it all together, we get the following: There is a path from
  $s$ to $t$ in $G$ if, and only if, there is a common subsequence of
  length $m-1$ in the constructed strings: If there is a path, the
  sequence of the edges on it form a subsequence; and if there is such
  a subsequence, because of its length, it must contain exactly one
  edge leaving from each layer except the last -- and these edges must
  form a path as we just argued.
\end{proof}

Although we do not prove this, we remark that
$\Class{NL}$-completeness already holds for $|S| = 3$, while for
$|S| = 2$ the complexity appears to drop significantly.

\begin{corollary}\label{corollary-lcs-injective}
  $\PLang[strings]{lcs-injective}$ is para-L-complete for
  $\Para\Class{NL}$. 
\end{corollary}

\begin{proof}
  The problem lies in $\Para\Class{NL}$ since by
  Theorem~\ref{theorem-lcs-injective} we can solve any instance in
  $\Class{NL}$ without even using the parameter. On the other hand,
  the theorem also shows that a slice of the parameterized problem
  (namely for 4 strings) is already hard for $\Class{NL}$. It is a
  well-known fact that in this case the parameterized problem is hard
  for the corresponding para-class, which happens to be
  $\Para\Class{NL}$. 
\end{proof}

\begin{theorem}\label{theorem-lcs}
  $\PLang[strings]{lcs}$ is para-L-complete for $\Class
  N[f\mathrm{poly}, f\log]$. 
\end{theorem}

\begin{proof}
  Clearly, $\PLang[strings]{lcs} \in \Class N[f\mathrm{poly},
  f\log]$ since a nondeterministic machine can guess the common
  subsequence on the fly and only needs to keep track of $k$ pointers
  into the strings, which can be done in space $O(k \log_2 n)$.
  To prove hardness, we reduce from the $\Class N[f\mathrm{poly},
  f\log]$-complete problem $\PLang[cells]{dag-nca}$, the
  acceptance problem for nondeterministic cellular
  \textsc{dag}-automata, see Theorem~\ref{theorem-cellular}.
  
  Our first step is to tackle the problem that in an \textsc{lcs} instance we
  choose ``one symbol after the other'' whereas in a cellular
  automaton all cells make one step in parallel. To address this, we
  introduce a new intermediate problem $\PLang[cells]{dag-nca-sequential}$
  where the model of computation of the cellular automaton is modified
  as follows: Instead of all $k$ cells making one parallel step,
  initially only the first cell makes a transition, then the second
  cell makes a transition (seeing already the new state reached at the
  first cell, but still the initial state of the third cell), then the
  third cell (seeing the new state of cell two and the old of cell
  four), and so on up to the $k$th cell. Then, we begin again with the
  first cell, followed by the second cell, and so on.
  
  \begin{claim}
    $\PLang[cells]{dag-nca}$ reduces to  $\PLang[cells]{dag-nca-sequential}$.
  \end{claim}

  \begin{proof}[of the claim]
    The trick is to have cells
    ``remember'' the states they were in: On input of $(C,\penalty0q_1\dots q_k)$, we
    construct a ``sequential'' cellular automaton $C'$ as follows. If
    $Q$ is the state set of $C$, the state set of $C'$ is $Q \times Q$.
    Each state $q \in Q'$ is now a pair $(q^{\mathrm{previous}},
    q^{\mathrm{current}})$. The transition relation is adjusted as follows:
    If there used to be a transition
    $(q_{\mathrm{left}},q_{\mathrm{old}},q_{\mathrm{right}},q_{\mathrm{new}})
    \in Q^4$, meaning 
    that a cell of the parallel automaton~$C$ can switch to state $q_{\mathrm{new}}$ if
    it is in state $q_{\mathrm{old}}$, its left neighbor is in state
    $q_{\mathrm{left}}$, and its right neighbor is in state
    $q_{\mathrm{right}}$, the we now have the following 
    transitions in $C'$:
    $((q_{\mathrm{left}},x),(y,q_{\mathrm{old}}),(z,q_{\mathrm{right}}),(q_{\mathrm{old}},q_{\mathrm{new}}))$
      where 
    $x,y,z \in Q$ are arbitrary. Indeed, this transition will switch a
    cells state based on the \emph{previous} state of the cell before it
    and on the \emph{current} state of the cell following it and will
    store that previous state. For the first and last cells, this
    construction is adapted in the obvious manner. Clearly, the
    resulting sequential automaton will arrive in a sequence
    $(x_1,q_1)\dots(x_k,q_k)$ of states  for some $x_i \in Q$ after
    $t\cdot k$ steps if, and only if, the original automaton arrives in
    states $q_1\dots q_k$ after $t$~steps. This proves the reduction.
  \end{proof}

  \medskip\emph{The basic idea.}
  We now show how  $\PLang[cells]{dag-nca-sequential}$ can be reduced
  to $\PLang[strings]{lcs}$. Before we plunge into the details, let us
  first outline the basic idea: Each cell of a cellular
  \textsc{dag}-automaton ``behaves somewhat like a reachability
  problem'' meaning that we must find out whether the 
  automaton will arrive in the accepting state starting from the
  initial state. Thus, as in the proof of
  Theorem~\ref{theorem-lcs-injective}, we use four strings to
  represent a cell of the automaton, giving a total of $4k$ strings,
  where $k$ is the number of cells. However, the cells do not act
  independently; rather each step of a cell depends on the states of
  the two neighboring cells. Fortunately, this ``control'' effect can
  be modelled by adding an ``edge's'' symbol (actually, a transition's
  symbol) not only to the four strings of the cell, but also to the
  four strings of predecessor and successor cells at the right
  position (namely ``before the required state symbol''). 
  In the following, we explain the idea just outlined in detail.
  
  Let $(C,q_1\dots q_k)$ be given as input for the reduction. Since
  $C$ is sequential and also a \textsc{dag}-automaton, its steps can be
  grouped into at most $t$ many groups (``major steps'') of $k$
  sequential steps (``minor steps'') taken by cells $1$ to $k$ in that
  order, where $t$ depends linearly on the size of~$C$. By modifying
  $C$, if necessary, we may assume that $C$ makes exactly $t\cdot k$
  sequential steps when it accepts the input and, otherwise, makes
  strictly less steps. We use $s$ to denote a major step number.
  
  \medskip\emph{Construction of the strings.}
  We map $(C,q_1\dots q_k)$  to $4k$ strings $s_1^1$, $s_2^1$,
  $s_3^1$, $s_4^1$, \dots, $s_1^k$, $s_2^k$, $s_3^k$, $s_4^k$ and ask
  whether they have a common subsequence of length $t\cdot k$. Each
  group of four strings is setup similarly to the four strings from the 
  proof of Theorem~\ref{theorem-lcs-injective}: $s_1^i$ and $s_2^i$
  model the states (vertices) the $i$th cell has just before odd
  major steps~$s$; and  $s_3^i$ and $s_4^i$ model the states the cell
  has before even major steps~$s$.

  Consider cell $i$ and its four strings $s_1^i$ to $s_4^i$. Recall
  that in  Theorem~\ref{theorem-lcs-injective} we conceptually added
  the vertices of the first layer in opposite orders to $s_1^i$
  and~$s_2^i$, although in reality these vertices were not part of the 
  final strings and were added to make it easier to explain where the
  actual symbols (the edges) were placed in the strings. In our
  setting, the role of the vertices on the first layer is taken by the
  states $Q = \{q_1,\dots,q_n\}$ of the automaton $C$ tagged by the
  major step number~$1$. Thus, $s_1^i$ starts (conceptually) with
  $(q_1,1) \dots (q_n,1)$ and $s_2^i$ starts with $(q_n,1) \dots
  (q_1,1)$. Next come tagged versions of the states just before the
  third major step, so $s_1^i$ continues $(q_1,3) \dots (q_n,3)$ and
  $s_2^i$ with $(q_n,3) \dots (q_1,3)$. We continue in this way for
  all odd major steps. For even major steps, we add analogous strings
  to $s_3^i$ and $s_4^i$.

  Continuing the idea from Theorem~\ref{theorem-lcs-injective}, we now
  add ``edges'' to the strings. However, instead of an edge from one
  vertex so another, the transition relation of a cellular automaton
  contain 4-tuples $f =
  (f_{\mathrm{left}},f_{\mathrm{old}},f_{\mathrm{right}},f_{\mathrm{new}})
  \in Q^4$ of states, which allows a cell to switch to state
  $f_{\mathrm{new}}$ when it was in state $f_{\mathrm{old}}$ and its
  left neighbor was in state $f_{\mathrm{left}}$ and the right
  neighbor was in state~$f_{\mathrm{right}}$. Recall that in
  Theorem~\ref{theorem-lcs-injective}, for each $e$ from some vertex
  $a$ on an odd layer to a vertex $b$, we 
  added the symbol $e$ \emph{after} $a$ in the first two strings and
  \emph{before} $b$ in the last two strings. In a similar way, for the
  cellular automaton for each 4-tuple $f$ we add new ``symbols''
  $(f,s,i)$ consisting of a transition, a major step number, and a
  cell index~$i$ to the strings. This symbol is added at several
  places to the strings (we assume that $s$ is odd; for even $s$
  exchange the roles of the first two and the last two strings
  everywhere); sometimes even more than once. The rules are as follows:
  \begin{enumerate}
  \item Iterate over all $(f,s,i)$ in some order and insert $(f,s,i)$
    directly after $(f_{\mathrm{old}},s)$ in $s_1^i$. 
  \item Next, again iterate over all $(f,s,i)$, but now in reverse
    order, and insert $(f,s,i)$ after $(f_{\mathrm{old}},s)$ in $s_2^i$.
  \end{enumerate}
  Note that using the two opposite orderings, as in
  Theorem~\ref{theorem-lcs-injective}, for each $(f_{\mathrm{old}},s)$
  at most one $(f,s,i)$ can be part of a common subsequence.
  \begin{enumerate}\setcounter{enumi}{2}
  \item Next, iterate over all $(f,s,i)$ in some order and insert $(f,s,i)$
    directly before $(f_{\mathrm{new}},s+1)$ in $s_3^i$. 
  \item Next, iterate over all $(f,s,i)$ in reverse order and insert $(f,s,i)$
    directly before $(f_{\mathrm{new}},s+1)$ in $s_4^i$. 
  \end{enumerate}
  The effect of the above is to make the automaton  switch to
  $f_{\mathrm{new}}$ in cell~$i$ after major step~$s$. Now, we still
  need to ensure that 
  this switch is only possible when the preceding cell has already
  switched to state $f_{\mathrm{left}}$ after step~$s$ and the next
  cell is in state $f_{\mathrm{right}}$ before step~$s$. 
  \begin{enumerate}\setcounter{enumi}{4}
  \item Next, iterate over all $(f,s,i)$  and insert $(f,s,i)$
    directly after $(f_{\mathrm{left}},s+1)$ in $s_3^{i-1}$ and $s_4^{i-1}$. For
    $i=1$, no symbols are added.
  \item Next, iterate over all $(f,s,i)$  and insert $(f,s,i)$
    directly after $(f_{\mathrm{right}},s)$ in $s_1^{i+1}$ and $s_2^{i+1}$. For
    $i=k$, no symbols are added.
  \end{enumerate}
  Note that since the last two steps are applied later, the added
  symbols are ``nearer'' to the state symbols than the symbols added
  in the first two steps. In particular, a common subsequence can
  contain first a symbol added in step~6 added after some $(q,s+1)$, then a
  symbol added after $(q,s)$ in step~5, and then symbols added before
  or after $(q,s)$ in one of the first four steps.
  
  The last rule ensures that when a tuple $(f,s,i)$ is not
  mentioned for a string by one of the first six rules, we can always
  make it part of a common subsequence:
  \begin{enumerate}\setcounter{enumi}{6}
  \item Finally, iterate over the $4k$ strings. For each such string $s_j^i$,
    consider the set $X$ of all $(f,s,i)$ that are not present in
    $s_j^i$. Add all symbols of $X$ once after each letter of $s_j^i$.
  \end{enumerate}
  
  As the last step of the construction of the strings, in order to
  model the initial configuration $q_1\dots q_k$ of the automaton, for
  each $i \in \{1,\dots,k\}$ in $s_1^i$ to $s_4^i$ we remove all
  symbols before $(q_i,1)$.

  \medskip\emph{Correctness: First direction.}
  Having finished the description of the reduction, we now argue that
  it is correct. For this, first assume that the automaton, does,
  indeed, accept the input sequence $q_1\dots q_k$. By assumption,
  this means that the automaton will make $t \cdot k$ sequential
  steps. Assume that in major step~$s$ and minor step~$i$ the
  automaton makes transition $f^{s,i}$, meaning that the $i$th cell
  switches its state from $f^{s,i}_{\mathrm{old}}$ to
  $f^{s,i}_{\mathrm{new}}$.

  We claim that the sequence $(f^{1,1},1,1)\penalty0(f^{1,2},1,2)\dots\penalty0 
  (f^{1,k},1,k)\penalty0 (f^{2,1},2,1) \dots\penalty0
  (f^{t,k},t,k)$ is a common subsequence of all $s_j^i$.
  To see this,
  consider the first symbol $(f^{1,1},1,1)$. It will be present both
  in $s_1^1$ and $s_2^1$ since for the first transition the first cell was
  exactly in state $q_1 = f^{1,1}_{\mathrm{old}}$ and, thus, this
  symbol \emph{followed} $(q_1,1)$ in the construction and was not
  removed in the last construction step. The symbol is also present in
  $s_3^1$ and $s_4^1$, namely right before the (``virtual'') pair
  $(f_{\mathrm{new}},2)$. The symbol will also be present in $s_1^2$
  and $s_2^2$ since $q_2 = f^{1,1}_{\mathrm{right}}$ and we added
  $(f^{1,1},1,1)$ to both $s_1^2$ and $s_2^2$ in step~6. Finally, the
  symbol will be present in all other strings near the beginning
  because of step~7.
  
  Next, consider the second symbol $(f^{1,2},1,2)$, which corresponds to
  the second step the automaton has taken. Here, the second cell
  switches from $f^{1,2}_{\mathrm{old}}$ to  $f^{1,2}_{\mathrm{new}}$
  because the first cell has already switched to
  $f^{1,2}_{\mathrm{left}} = f^{1,1}_{\mathrm{new}}$ during the first
  transition and the third cell is still in $f^{1,2}_{\mathrm{right}}
  = q_3$. Now, observe that in all strings $(f^{1,2},1,2)$ does,
  indeed, come after $(f^{1,1},1,1)$: For $s_1^2$ to $s_4^2$ this is
  because of steps 1 to~4. For $s_3^1$ to $s_4^1$, we have, indeed, 
  $(f^{1,2},1,2)$ following $(f^{1,1},1,1)$ by step~5. For $s_1^3$ to
  $s_2^3$, the symbol $(f^{1,2},1,2)$ is present by step~6. All other
  strings contain the symbol by step~7 near the beginning.

  Continuing in a similar fashion with the other symbols, we see that
  the sequence
  \begin{align*}
    (f^{1,1},1,1)\dots(f^{t,k},t,k)    
  \end{align*}
  is a common subsequence of all
  strings and it clearly has length $t\cdot k$.

  \medskip\emph{Correctness: Second direction.}
  It remains to argue that if there is a common subsequence of the
  strings of length $t\cdot k$, then the automaton accepts the
  input. First observe that the common subsequence must be of the
  form  $(f^{1,1},1,1)\penalty0(f^{1,2},1,2)\dots\penalty0 
  (f^{1,k},1,k)\penalty0 (f^{2,1},2,1) \dots\penalty0
  (f^{t,k},t,k)$. The reason is that for any two symbols
  $(f,s,i)$ and $(f',s',i')$ if $s <s'$ then the first of these
  symbols always comes before the second in all strings. The same is
  true if $s = s'$ and $i < i'$. Finally, for $s=s'$ and $i=i'$, the
  opposite orderings for the symbols in steps  1 and~2 (and, also, in
  steps 3 and~4) ensure that at most one of the two symbols can be
  present in a common subsequence. Thus, the indices stored in the
  symbols of the common subsequence must strictly increase and, since
  the length of the sequence is $t\cdot k$, all possible indices must
  be present.

  We must now argue that the $f^{s,i}$ form a sequence of transitions
  that make the automaton accept. For this, we perform an induction
  on the length of an initial segment up to some symbol
  $(f^{s_0,i_0},s_0,i_0)$ of the common sequence. For each
  cell index~$i$, let $f^i = (f^{s,i},s,i)$ be the last symbol in the
  segment whose last component is $i$. Let $q^i =
  f^i_{\mathrm{new}}$ or, if the segment is so short that there is no
  $f^i$, let $q^i$ be the initial state~$q_i$. The inductive
  claim is that after $(s_0-1) \cdot k + i_0$ steps of the
  automaton, the cells will have reached exactly states
  $q^1,\dots,q^k$. Clearly, this is correct at the start. For the
  inductive step, the crucial observation is that steps 1 to~6
  guarantee that for $i_0 < k$ the only symbol
  $(f^{s_0,i_0+1},s_0,i_0+1)$ that can follow $(f^{s_0,i_0},s_0,i_0)$
  in a common sequence is one that makes that cell $i_0+1$ change its
  state according to the transition $f^{s_0,i_0+1}$. For $i_0 = k$, we
  similarly have that only symbols $(f^{s_0+1,1},s_0+1,1)$ can follow
  that make cell $1$ change its state according to the transition
  $f^{s_0+1,1}$. 
\end{proof}

\section{Conclusion}

Bounded nondeterminism plays a key role in parameterized complexity
theory since it lies at the heart of the definition of important
classes like $\Class W[\Class P]$, but also of $\Class W[1]$. In the
present paper we introduced a ``W-operator'' that cannot only be
applied to $\Class P$, yielding $\Para[W]\Class P$, but also to
classes like $\Class{NL}$ or $\Class{NC}^1$. We showed that 
``union versions'' of problems complete for $\Class P$,
$\Class{NL}$, and $\Class L$ tend to be complete for $\Para[W]\Class
P$, $\Para[W]\Class{NL}$, and $\Para[W]\Class L$. Several important
problems studied in parameterized complexity turn out to be union
problems, including $\PLang{circuit-sat}$ and
$\PLang{weighted-sat}$, and we could show that the latter problem is
complete for $\Para[W]\Class{NC}^1$. For the associative generability problem
$\PLang{agen}$, which is also a union problem, we established its
$\Para[W]\Class{NL}$-completeness. An interesting open problem is 
determining the complexity of the ``universal'' version of
$\Lang{agen}$, where the question is whether \emph{all}
size-$k$ subsets of the universe are generators. Possibly, this
problem is complete for $\Para[W_\forall]\Class{NL}$. 

We showed that different problems are complete for
the time--space class $\Class N[f \operatorname{poly}, f\log]$.
We shied away from presenting complete problem for the classes
$\Class D[n^f, f \operatorname{poly}]$ and $\Class N[n^f,\penalty0 f
\operatorname{poly}]$ because in their definition we need
restrictions like ``the machine may make at most $n^k$ steps where $k$
is the parameter.'' Such artificial parameterizations have been
studied, though: In \cite[Theorem 2.25]{FlumG2006} Flum and Grohe show that
``$\PLang{exp-dtm-halt}$'' is complete for $\Class{XP}$. Adding a
unary upper bound on the number of steps to the definition of the
problem yields a problem easily seen to be complete for $\Class D[n^f,
f \operatorname{poly}]$. Finding a \emph{natural} problem complete for
the latter class is, however, an open problem. 

\subsubsection*{Acknowledgements.} We would like to thank Michael
Elberfeld for helping us with the proof of Theorem~\ref{theorem:agen}.

\bibliographystyle{plain}
\bibliography{main}

\end{document}